\documentclass[twocolumn,superscriptaddress,aps,showpacs,amsmath,amstex,amssymb,citeautoscript,floatfix,pra,preprintnumbers]{revtex4-2}
\pdfoutput=1
\usepackage[english]{babel}

\usepackage{letltxmacro}
\usepackage{latexsym}
\usepackage{booktabs}  
\usepackage{siunitx}   
\usepackage{multirow}  
\usepackage{algpseudocode}

\LetLtxMacro{\ORIGselectlanguage}{\selectlanguage}
\makeatletter
\DeclareRobustCommand{\selectlanguage}[1]{%
  \@ifundefined{alias@\string#1}
    {\ORIGselectlanguage{#1}}
    {\begingroup\edef\x{\endgroup
       \noexpand\ORIGselectla   nguage{\@nameuse{alias@#1}}}\x}%
}
\newcommand{\definelanguagealias}[2]{%
  \@namedef{alias@#1}{#2}%
}
\makeatother

\makeatletter
\def\maketitle{
\@author@finish
\title@column\titleblock@produce
\suppressfloats[t]}
\makeatother

\definelanguagealias{en}{english}
\definelanguagealias{English}{english}
\usepackage{graphicx}
\usepackage{amsmath}
\usepackage{amsfonts}
\usepackage{amssymb}
\usepackage{amsthm}
\usepackage{cancel}
\usepackage{mathtools}

\usepackage{bm}
\usepackage{color}
\usepackage[percent]{overpic}
\usepackage{soul} 
\usepackage{amssymb}
\usepackage{wasysym}
\usepackage{dsfont}
\usepackage{float}
\usepackage{physics}
\usepackage{hyperref}
\usepackage{comment}
\usepackage{enumitem}
\usepackage{booktabs}
\usepackage{textgreek}
\usepackage{tcolorbox}

\newcounter{alg}

\hypersetup{
  colorlinks   = true, 
  urlcolor     = blue, 
  linkcolor    = blue, 
  citecolor   = red 
}
\usepackage[mathscr]{euscript}

\definelanguagealias{en}{english}
\definelanguagealias{English}{english}
\usepackage{graphicx}
\usepackage{amsmath}
\usepackage{amsfonts}
\usepackage{amssymb}
\usepackage{amsthm}
\usepackage{cancel}
\usepackage{mathtools}
\usepackage{bm}
\usepackage{color}
\usepackage[percent]{overpic}
\usepackage{soul} 
\usepackage{amssymb}
\usepackage{wasysym}
\usepackage{dsfont}
\usepackage{float}
\usepackage{physics}
\usepackage{comment}
\usepackage{enumitem}
\usepackage{booktabs}
\usepackage{tcolorbox}
\usepackage{hyperref}

\newtheorem{theorem}{Theorem}
\newtheorem{result}{Result}
\newtheorem{lemma}[theorem]{Lemma}
\newtheorem*{lemma*}{Lemma}

\usepackage{tcolorbox}
\usepackage{hyperref}
\usepackage{cleveref}

\renewcommand{\thealg}
{\arabic{alg}}

\newtcolorbox[use counter=alg,
              crefname={algorithm}{algorithms},
              Crefname={Algorithm}{Algorithms}]
{alg}[2][]{%
  floatplacement=#1,         
  float,                     
  colback=cyan!5!white,      
  colframe=cyan!50!black,    
  colbacktitle=cyan!85!black,
  fonttitle=\bfseries,       
  title=Algorithm~\thealg: #2
}

\hypersetup{
  colorlinks   = true, 
  urlcolor     = blue, 
  linkcolor    = blue, 
  citecolor   = red 
}
\usepackage[mathscr]{euscript} 

\usepackage{varwidth} 
\usepackage{comment}

\usepackage{algorithm}
\usepackage{float}

\newtheorem*{theorem*}{Theorem}

\newtheorem{thm}{\protect\theoremname}
\theoremstyle{plain}

\theoremstyle{plain}

\theoremstyle{plain}
\newtheorem*{lem*}{\protect\lemmaname}
\theoremstyle{plain}
\newtheorem*{thm*}{\protect\theoremname}
\theoremstyle{plain}

\theoremstyle{plain}

\newtheorem{defn}{Definition}

\newtheorem{fact}[thm]{Fact}

\usepackage{graphicx}
\usepackage{}

\usepackage{verbatim}
\usepackage[normalem]{ulem}

\newcommand{\EPR}{\mathrm{EPR}}
\newcommand{\eq}[1]{\begin{equation}#1\end{equation}}
\newcommand{\eqs}[1]{\begin{equation}\begin{split}#1\end{split}\end{equation}}

\newcommand{\figref}[1]{Fig.\,\ref{#1}}

\newcommand{\THUIIIS}{Center for Quantum Information, IIIS, Tsinghua University, Beijing, China}
\newcommand{\VT}{Department of Computer Science, Virginia Tech, Alexandria, VA 22314, USA}
\newcommand{\Phasecraft}{Phasecraft Inc., Washington DC, USA}
\newcommand{\DukeMath}{Department of Mathematics, Duke University, Durham, NC 27708, USA}
\newcommand{\DukeECE}{Department of Electrical and Computer Engineering, Duke University, Durham, NC
27708, USA}
\newcommand{\DukeQuantumCenter}{Duke Quantum Center, Duke University, Durham, NC
27701, USA}
\newcommand{\Harvard}{Department of Physics, Harvard University, Cambridge, MA 02138, USA}
\newcommand{\HarvardCS}{School of Engineering and Applied Sciences, Harvard University, Allston, Massachusetts 02134, USA}
\newcommand{\Caltech}{Institute for Quantum Information and Matter, California Institute of Technology, CA 91125, USA}

\begin{document}
\title{
Ansatz-free Hamiltonian learning with Heisenberg-limited scaling
}

\author{Hong-Ye Hu}
\altaffiliation{Alphabetical ordering and equal contributions. \\
\href{mailto:hongyehu@g.harvard.edu}{hongyehu@g.harvard.edu}}
\affiliation{\Harvard}
\author{Muzhou Ma}
\altaffiliation{Alphabetical ordering and equal contributions. \\
\href{mailto:muzhouma2002@gmail.com}{muzhouma2002@gmail.com}}
\affiliation{\Caltech}
\author{Weiyuan Gong}
\affiliation{\HarvardCS}
\author{Qi Ye}
\affiliation{\THUIIIS}
\affiliation{\HarvardCS}
\author{Yu~Tong}
\affiliation{\DukeMath}
\affiliation{\DukeECE}
\affiliation{\DukeQuantumCenter}
\author{Steven T. Flammia}
\affiliation{\VT}
\affiliation{\Phasecraft}
\author{Susanne F. Yelin}
\altaffiliation{\href{mailto:syelin@g.harvard.edu}{syelin@g.harvard.edu}}
\affiliation{\Harvard}

\begin{abstract}
Learning the unknown interactions that govern a quantum system is crucial for quantum information processing, device benchmarking, and quantum sensing. The problem, known as Hamiltonian learning, is well understood under the assumption that interactions are local, but this assumption may not hold for arbitrary Hamiltonians. Previous methods all require high-order inverse polynomial dependency with precision, unable to surpass the standard quantum limit and reach the gold standard Heisenberg-limited scaling. Whether Heisenberg-limited Hamiltonian learning is possible without prior assumptions about the interaction structures, a challenge we term \emph{ansatz-free Hamiltonian learning}, remains an open question. In this work, we present a quantum algorithm to learn arbitrary sparse Hamiltonians without any structure constraints using only black-box queries of the system's real-time evolution and minimal digital controls to attain Heisenberg-limited scaling in estimation error. Our method is also resilient to state-preparation-and-measurement errors, enhancing its practical feasibility. We numerically demonstrate our ansatz-free protocol for learning physical Hamiltonians and validating analog quantum simulations, benchmarking our performance against the state-of-the-art Heisenberg-limited learning approach. Moreover, we establish a fundamental trade-off between total evolution time and quantum control on learning arbitrary interactions, revealing the intrinsic interplay between controllability and total evolution time complexity for any learning algorithm. These results pave the way for further exploration into Heisenberg-limited Hamiltonian learning in complex quantum systems under minimal assumptions, potentially enabling new benchmarking and verification protocols.
\end{abstract}
\maketitle

\section{Introduction}
Understanding the interactions that govern nature is a central goal in physics. In quantum systems, these interactions are described by the Hamiltonian, which dictates both the static and dynamic properties of the system. Consequently, given access to a quantum system with an unknown Hamiltonian, a fundamental question arises: what is the most efficient method to learn the interactions of such a system? While this question underpins much of quantum many-body physics, it has become increasingly relevant in practice due to the remarkable progress in quantum science and technology, notably the emergence of programmable analog quantum simulators \cite{Daley23,ProbingQuantumDynamics,SpinLiquid} and early fault-tolerant quantum computers \cite{PhysRevX.11.041058,demo_FTgates,LogicalRydberg,GoogleBreakEven,AlgorithmicFT}. These platforms promise to simulate complex quantum phenomena that remain intractable with classical computation. Nevertheless, they also introduce a pressing challenge: how to rigorously validate and benchmark such engineered quantum devices \cite{QuantumVerification}. Therefore, Hamiltonian learning is a critical tool to not only probe the unknown interactions but also characterize and control these engineered quantum systems \cite{RobustLearningSuperConducting,HuangTongFangSu2023learning,ma2024learningkbodyhamiltonianscompressed,LiTongNiGefenYing2023heisenberg,QLFH,Bakshi_2024,PhysRevA.110.062421,RobustEstimation,Yu2023robustefficient,SampleEfficientLearning,EntanglementHamiltonianLearning,Olsacher_2025,2024arXiv240101308O,PauliTransfer,MobusBluhmCaroEtAl2023dissipation}.  Beyond programmable quantum simulators, Hamiltonian learning also arises in quantum metrology and sensing, where one aims to determine an unknown field (i.e., the Hamiltonian) to precision $\epsilon$ at the so-called $\mathcal{O}(1/\epsilon)$ Heisenberg limit. Refining these learning strategies will not only enable the certification of next-generation quantum hardware but also open new avenues in precision sensing and the broader landscape of quantum technologies.

Traditional methods for Hamiltonian learning often rely on preparing either an eigenstate or the thermal (Gibbs) state of the underlying Hamiltonian \cite{HighTempGibbs,2023arXiv231002243B,PracticalGibbsLearning,Qi2019determininglocal,PhysRevLett.124.160502,BayesianHamiltonianLearning}. The coefficients of the unknown Hamiltonian are determined by solving a system of polynomial equations involving the expectation values of numerous Pauli observables. However, the state preparation step is non-trivial, and these methods are constrained by the so-called standard quantum limit, where achieving a precision $\epsilon$ in the learned coefficients requires a total experimental time scaling as $\mathcal{O}(1/\epsilon^2)$.

\begin{figure*}[htbp]
    \includegraphics[width=1.0\linewidth]{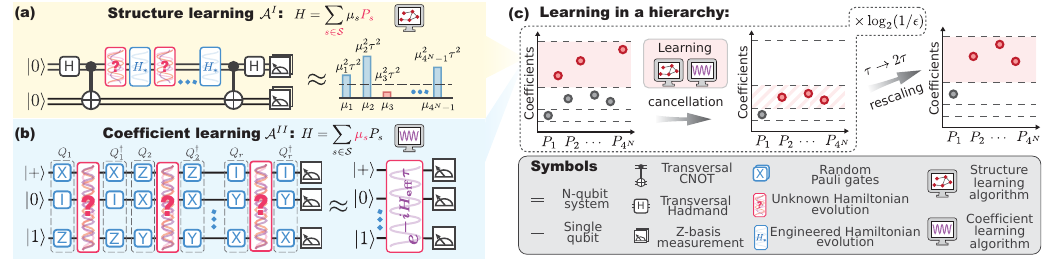}
    \caption{(a) Quantum circuit for the structure-learning subroutine $\mathcal{A}^I$. It prepares n pairs of 2-qubit Bell states between the original and ancillary systems via transversal gates. The original system then evolves coherently under the unknown Hamiltonian $H$ and the engineered Hamiltonian $H_*$, where $H_*$ consists of the large terms in $H$ learned in previous steps with an opposite sign. The combined system is then measured on the n-pair Bell basis. Nontrivial outcomes have probabilities proportional to $\mu_i^2$, enabling inference of the interaction structure. (b) Quantum circuit for the coefficient-learning subroutine $\mathcal{A}^{II}$. By inserting random Pauli gates from a designed set into the unknown Hamiltonian’s evolution, the time evolution of a specific term is approximated, allowing the interaction strength $\mu_s$ to be extracted via robust frequency estimation. (c) Combining these subroutines enables hierarchical coefficient estimation, achieving Heisenberg-limited scaling.\label{fig:protocol}}
\end{figure*}

Recently, inspired by quantum metrology, a new class of Hamiltonian learning algorithms has been proposed that achieves Heisenberg-limited scaling \cite{HuangTongFangSu2023learning,ma2024learningkbodyhamiltonianscompressed,LiTongNiGefenYing2023heisenberg,PhysRevA.110.062421,QLFH,Bakshi_2024}. These methods require only simple initial state preparation and black-box queries of the Hamiltonian dynamics. Despite their efficiency, these approaches require a crucial assumption that the interactions are either geometrically local or $k-$local. However, in many scenarios, the exact interaction structure is not known in advance, allowing for potentially arbitrary interactions. Consequently, the search space for the unknown Hamiltonian structure becomes exponentially large, making it challenging to identify interaction terms. Moreover, the possibility of non-commuting terms further complicates the accurate estimation of each coefficient. Efforts to extend these existing methods to arbitrary Hamiltonians have encountered significant obstacles: some approaches demand highly complex quantum controls, such as block encoding and the time reversal evolutions \cite{2024arXiv241021635Z}, while others fail to reach the optimal Heisenberg-limited scaling \cite{2024arXiv241100082A}. Therefore, it remains a fundamental open question whether one can achieve Heisenberg-limited Hamiltonian learning with only simple black-box queries to the unitary dynamics and no prior assumptions of the interaction structure—a task we refer to as \emph{ansatz-free Hamiltonian learning}.

In this work, we propose a novel Hamiltonian learning algorithm that overcomes these limitations. 
Our method achieves Heisenberg-limited scaling for arbitrary Hamiltonians in the following qubit form, including non-local ones, with the total experimental time scaling polynomially with the number of Pauli terms in the Hamiltonian. 
To be more explicit, any $n$-qubit Hamiltonian $H$ can be expressed as:
\begin{equation}
\label{eq:unknown_Hamiltonian}
    H = \sum_{s\in \mathcal{S}} \mu_s P_s,
\end{equation}
where $P_s$ are the $n$-qubit traceless Pauli operators, $\mathcal{S}$ is the set of Pauli operators that constitute $H$ and $\mu_s$ are the unknown coefficients. Since the coefficients $\mu_s$ can always be rescaled by the evolution time, we normalize them as $|\mu_s| \leq 1$, following standard practice in the analysis of Hamiltonian learning complexity \cite{HuangTongFangSu2023learning,HighTempGibbs}. We use $M$ to denote the number of Pauli terms with nonzero coefficients in $H$. Importantly, we impose no structural constraints on $H$: the Pauli terms $P_s$ can be nonlocal and can even have support on up to $\mathcal{O}(n)$ qubits. Moreover, any fermionic Hamiltonian can be mapped into the form of \Cref{eq:unknown_Hamiltonian}, allowing our results to apply directly to fermionic systems as well.  This ansatz-free characterization sets our approach apart from previous works in a fundamental way.

Without prior knowledge of which terms are in the Hamiltonian and what the coefficient values are, there are two specific difficulties for learning such a Hamiltonian: 1. what is the \emph{structure} $\mathcal{S}$, the $M$ Pauli terms, of this Hamiltonian; 2. what are the \emph{coefficients}  $\mu_s$ with respect to each term in this Hamiltonian. We tackle these difficulties by 
cycling through an alternating hierarchy of two steps: structure learning and coefficient learning.
In the structure learning phase, we identify the Pauli terms with large coefficients by directly sampling a simple quantum circuit, where terms with large coefficients will dominate the outcome.
In the coefficient learning phase, we isolate the identified terms by applying ensembles of single-qubit Pauli gates, similar to techniques such as dynamical decoupling \cite{ReviewOfDD,PhysRevX.10.031002,PhysRevLett.119.183603,SDD} or Hamiltonian reshaping \cite{HuangTongFangSu2023learning,ma2024learningkbodyhamiltonianscompressed}. We then estimate the coefficients through robust frequency estimation \cite{KimmelLowYoder2015robust}. Our algorithm not only achieves Heisenberg-limited scaling in terms of total experimental time but is also resilient to state-preparation-and-measurement (SPAM) errors.

To the best of our knowledge, this is the \emph{first} quantum algorithm capable of learning arbitrary Hamiltonians with Heisenberg-limited scaling using only product state inputs, single-qubit measurements, and black-box access to the Hamiltonian dynamics. This work not only resolves a long-standing theoretical question about Hamiltonian learning but also introduces a practical algorithm with minimal experimental requirements.

\section{The learning protocol}

We consider the task of learning an unknown many-body Hamiltonian H, as given in \Cref{eq:unknown_Hamiltonian}, through black-box access to its time evolution $e^{-iHt}$ for arbitrary $t$, and a programmable quantum computer. A natural figure of merit for evaluating the efficiency of any learning protocol is the total experimental time $T$. Suppose a protocol uses $J$ experiments to estimate the Hamiltonian coefficients $\bm{\mu}$ to precision $\epsilon$, where the $j$th experiment involves an evolution of duration $t_j$ and is repeated $L_j$ times. Then, the total experimental time is
\eq{T = \sum_{j=1}^{J} L_j t_j.\label{eq:total_time_main}}
A protocol achieves Heisenberg-limited scaling if $T = \mathcal{O}(1/\epsilon)$, and the standard quantum limit if $T = \mathcal{O}(1/\epsilon^2)$. In what follows, we describe in greater detail how our protocol alternates between structure learning and coefficient estimation in a hierarchical manner, with a total evolution time reaching the Heisenberg-limited scaling.

In the structure-learning step, we identify the dominant interaction terms by determining the support of the coefficient vector $\boldsymbol{\mu}=(\mu_1, \mu_2, \ldots, \mu_{4^n-1})^{T}$. 
To achieve this, we introduce two approaches: one (denoted as $\mathcal{A}^I$) employs $n$ pairs of 2-qubit Bell states shared between the original system and an ancillary system of the same size, while the other (denoted as $\mathcal{A}^{I'}$) uses only product state inputs and single-qubit measurements, eliminating the need for ancillary systems at the cost of a moderate increase in $M$-dependence. 
In the coefficient learning  (denoted as $\mathcal{A}^{II}$) step, we estimate the coefficients of Pauli operators identified in the preceding structure-learning step with robust frequency estimation.
Specifically, we first determine all the Pauli operators with coefficients $1/2< |\mu_i|\leq 1$, and learn their coefficients $\mu_i$ with $\mathcal{A}^I$ and $\mathcal{A}^{II}$. We then repeat those two steps for smaller coefficient ranges $1/4< |\mu_i|\leq 1/2$ and so on. In the $k$-th iteration, we learn coefficients that are $(1/2)^k< |\mu_s|\leq (1/2)^{k-1}$, continuing until $k=\lceil\log_2( 1/\epsilon )\rceil$, where $\epsilon$ is the desired learning precision. This hierarchical learning strategy achieves the gold standard Heisenberg-limited scaling, requiring a total experimental time having $1/\epsilon$ dependence up to a polylogarithmic factor to reach $\epsilon$-learning accuracy.

In the follows, we use $\widetilde{\mathcal{O}}(f)$ to omit $\text{polylog}(f)$ scaling factors. The main results of the hierarchical learning algorithm are summarized as follows: 

\begin{result}[Informal version of \Cref{thm:2-copy_learning_alg}]
    There exists a quantum algorithm for learning the unknown Hamiltonian $H$ as in  \Cref{eq:unknown_Hamiltonian} taking $n$ pairs of $2$-qubit Bell state as input for each experiment instance, querying to real-time evolution of $H$, and performing Bell-basis measurements that outputs estimation $\hat{\boldsymbol{\mu}}$ such that 
    \begin{equation}
    \label{eq:l-infty_norm}
        ||\hat{\boldsymbol{\mu}}-\boldsymbol{\mu}||_\infty \leq \epsilon
    \end{equation}
    with high probability. The total experimental time is
    \begin{equation}
        T = \widetilde{\mathcal{O}}\left(M^2/\epsilon\right).
    \end{equation}
    This algorithm has trivial classical post-processing and is robust against SPAM errors.
\end{result}

\begin{result}[Informal version of \Cref{thm:1-copy_learning_alg}]
    There exists an ancilla-free quantum algorithm for learning the unknown Hamiltonian $H$ as in  \Cref{eq:unknown_Hamiltonian} with product-state input, queries to real-time evolution of the Hamiltonian, and single-qubit measurements that outputs $\hat{\boldsymbol{\mu}}$ achieving \eqref{eq:l-infty_norm} with high probability. The total experimental time $T$ is 
    \begin{equation}
        T' = \widetilde{\mathcal{O}}\left(M^3\log(n)/\epsilon\right).
    \end{equation}
    This algorithm needs classical post-processing with time
    \begin{equation}
        T^C = \widetilde{\mathcal{O}}(M^5n\log(n))
    \end{equation}
    and is robust against a restricted class of SPAM errors~\cite{Flammia2021paulierror}.
\end{result}

In the following, we will outline the proof for both results by introducing the algorithms on structure learning ($\mathcal{A}^I$ and $\mathcal{A}^{I'}$) and coefficients learning ($\mathcal{A}^{II}$), and the corresponding proof ideas.

\subsection{Structure-learning algorithm}
We first consider the structure-learning algorithm $\mathcal{A}^I$, which takes as input $n$ pairs of Bell states shared between the original system and an $n$-qubit ancillary register. As a warm-up, consider the case where all unknown Hamiltonian coefficients are bounded away from zero, i.e., $\mu_m := \min \mu_s = \mathcal{O}(1)$. A key challenge in identifying the support $\mathcal{S}$ of the Hamiltonian arises from the fact that the constituent Pauli operators generally do not commute.
To address this, we combine the Bell sampling circuit with the coherent evolution driven by $H = \sum_{s}\mu_s P_s$ for time $\tau$. We prepare $n$ pairs of $2$-qubit Bell states $\ket{\Phi^{+}}^{\otimes n}= \left(\frac{\ket{00}+\ket{11}}{\sqrt{2}}\right)^{\otimes n}$ with transversal Hadamard and CNOT gates. The first qubit of each Bell pair undergoes evolution under $H$ for time $\tau$, while the second remains idle. A final Bell-basis measurement is performed by reapplying CNOT and Hadamard gates. The full circuit is illustrated in \figref{fig:protocol}(a), where the blue box labeled $H_*$ can be ignored for now. This setup enables direct sampling from the distribution $p(s|\bm{\mu})$, which is intrinsically linked to the support of $\bm{\mu}$. 

To illustrate the mechanism, consider a single-qubit Hamiltonian $H = \mu_x X + \mu_y Y + \mu_z Z$. Applying $H$ for a short time $\tau$ to the first qubit of a Bell pair yields:
\eqs{
&\ket{\Phi^{+}}-i\tau(\mu_x XI+\mu_y YI+\mu_z ZI)\ket{\Phi^{+}}+\mathcal{O}(\tau^2)\\
&=\ket{\Phi^{+}}-i\tau \mu_x\ket{\Psi^{+}}-i\tau\mu_y\ket{\Psi^{-}}-i\tau \mu_z \ket{\Phi^{-}}+\mathcal{O}(\tau^2),\label{eq:single_qubit_bell}
}
where $\ket{\Phi^{\pm}}$ and $\ket{\Psi^{\pm}}$ are the four Bell basis states. A Bell-basis measurement then yields outcomes with probabilities proportional to $\mu_x^2$, $\mu_y^2$, and $\mu_z^2$. For example, outcome $(1,0)$ corresponds to $\ket{\Psi^+}$ and occurs with probability proportional to $\mu_x^2$, and so on. Thus, the presence of any non-trivial Bell-basis measurement outcome directly implies that the corresponding coefficient is non-zero. This quantum circuit effectively performs importance sampling, using quantum coherence to extract structural information about the Hamiltonian. While it is, in principle, possible to estimate the actual values of $\mu_x$, $\mu_y$, and $\mu_z$ by learning the corresponding outcome probabilities, this would require many samples to build up accurate statistics and would not achieve Heisenberg-limited scaling. Instead, the goal of the structure-learning phase is simply to sample each non-trivial outcome \emph{at least once}, thereby identifying the non-zero components of the Hamiltonian. The precise values of the coefficients are then learned in the subsequent coefficient-learning step.

More specifically, by setting $\tau = \mathcal{O}((M\mu_m)^{-1})$, we show the outcome distribution can be lower-bounded as $p(s|\boldsymbol{\mu}) > \Omega\left(\mu_m^4/M^2\right)$ (see \Cref{app:structure_learning} for details). Consequently, by applying a union bound, one can ensure that all terms in the support set $\mathcal{S}=\{s:|\mu_s|\geq \mu_m\}$ are sampled at least once with high probability by repeating this quantum circuit $\mathcal{O}(\log(M)/\mu_m^2 \tau^2)$ times. The total evolution time under $H$ then scales as $\mathcal{O}(\log(M)/\mu_m^2 \tau)=\mathcal{O}(M\log(M)/\mu_m^3)$. 

Due to the finite-time evolution, different terms in $\mathcal{S}$ get multiplied together in second and higher orders as a result of Taylor expansion, leading to the possibility of spurious contributions. In rare cases, this can result in a false-positive detection of a Pauli string s with vanishing or negligible coefficient $\mu_s$. For instance, in the single-qubit example, even if $\mu_x = 0$, a Bell-basis measurement might still yield the readout of state $|\Psi^{+}\rangle$, arising from higher-order terms like $(\mu_y \tau Y)(\mu_z \tau Z)$ in $\mathcal{O}(\tau^2)$ in \Cref{eq:single_qubit_bell}, even though they are orders of magnitude smaller. However, such terms are easily ruled out in the subsequent coefficient learning step, which accurately estimates their amplitudes and filters out any below the desired threshold $\mu_m$. These false-positive detections are also illustrated in the Application section, for example in \figref{fig:example1}(a). Importantly, the worst-case number of such false-positive events is also upper bounded by the number of samples $\mathcal{O}(M^2\log(M))$. 

To aid readers, we also provide a high-level pseudocode in Algorithm~\hyperref[alg:algorithm_I]{1}. The inputs to this subroutine are: (1) black-box access to the time evolution generated by the unknown Hamiltonian H; (2) an approximately learned partial Hamiltonian $\hat{H}$; (3) the coefficient threshold $\mu_m$; (4) the number of terms in the Hamiltonian $M$; and (5) the number of measurements $\mathfrak{M} = \mathcal{O}(1)$. Notably, the second input $\hat{H}$ plays a central role in the hierarchical learning framework: previously learned components of the Hamiltonian are used to cancel their contribution in future queries to H, as will be detailed in later sections.

\begin{alg}[htbp]{Structure learning subroutine $\mathcal{A}^{I}$ (two-copy)}
\label{alg:algorithm_I}
\textbf{Input}: \vspace{-2mm}
\begin{enumerate}
    \item Unknown Hamiltonian $H$ with black-box access to its evolution $U_{H}(t)=e^{-iHt}$\vspace{-2mm}
    \item Approximately learned part of the Hamiltonian $\hat{H}$\vspace{-2mm}
    \item Coefficient threshold $\mu_m$\vspace{-2mm}
    \item An estimated number of terms in the Hamiltonian $M$\vspace{-2mm}
    \item Number of measurements $\mathfrak{M}$
\end{enumerate}\vspace{-2mm}

\textbf{Output}: Estimated support set $\mathcal{S} = \{ P_s : |\mu_s| \geq \mu_m\}$\vspace{2mm}

\textbf{Pseudo-code}:\vspace{-2mm}
\begin{enumerate}
    \item Initialize the set of measurements: $\mathcal{B} \gets \emptyset$\vspace{-2mm}
    \item Set evolution time $\tau \gets \Theta((M\mu_m)^{-1})$\vspace{-2mm}
    \item Set Trotterization steps $r_1 \gets \Theta(M^2/\mu_m^2)$ \vspace{-2mm}
    \item \textbf{For} $i = 1$ \textbf{to} $\mathfrak{M}$ \textbf{do}:\vspace{-2mm}
    \begin{enumerate}
        \item Prepare $N$ Bell pairs: $\ket{\Phi^+}^{\otimes N}$ \vspace{-0mm}
        \item Simulate evolution: \\
        $|\psi\rangle \gets e^{-i(H-\hat{H})\otimes I^{\otimes N} \cdot \tau} \ket{\Phi^+}^{\otimes N}$\\
        by interleaving black-box queries of $U_H(\tau/r_1)$ with application of $e^{-i\hat{H}\tau/r_1}$ via Trotter–Suzuki decomposition
        \item Perform Bell-basis measurement on $|\psi\rangle$\\
        $\mathcal{B} \gets \mathcal{B} \cup \{b_i\}$ \quad (store measurement outcome)
    \end{enumerate}
    \vspace{-3mm}
    \item Post-process: Return support set $\mathcal{S}$ corresponding to frequently observed Pauli terms in $\mathcal{B}$
\end{enumerate}
\end{alg}

It is natural to ask whether the entanglement in the structure learning step is necessary. Surprisingly, we provide a negative answer to this question by showing an alternative algorithm with product state input and single-qubit measurement at the cost of a moderate increase in $M$ dependence. The key observation is that applying random Pauli gates before and after the evolution $e^{-iH\tau}$ will transform it into an effective Pauli channel, a process called Pauli twirling \cite{RC,SymmetrizedNoise,PauliNoise,PEC,PEC2}. In the single-qubit instance, the effective channel after applying random Pauli gates is
\eqs{
\Lambda_P(\rho)&=\underset{\sigma_T\sim \mathbb{P}}{\mathbb{E}}\sigma_T^{\dagger}e^{-iH\tau}\sigma_T \rho\sigma_T^{\dagger}e^{iH\tau}\sigma_T\\
&\approx \rho+\mu_x^2\tau^2 X\rho X+\mu_y^2\tau^2 Y\rho Y+\mu_z^2\tau^2 Z\rho Z,
}
which is a Pauli channel with Pauli error rates $\mu_i^2\tau^2$. It can be generalized to multi-qubit systems by applying random Pauli gates on each qubit. Employing the Pauli error rates estimation protocol using product state inputs and single-qubit measurements \cite{Flammia2021paulierror}, we obtain an alternative approach $\mathcal{A}^{I'}$ for structure learning. In \Cref{sec:single-copy}, we provide the details of this approach.

\subsection{Coefficient-learning algorithm}
\label{sec:Coefficient-learning algorithm}
The coefficient-learning algorithm $\mathcal{A}^{II}$ takes the structure $\mathcal{S}$ as input and outputs the coefficients $\{\mu_s:s\in \mathcal{S}\}$. The key idea is to isolate each term $\mu_s P_s$ during the evolution and learn each individual coefficient $\mu_s$. The first step can be achieved by \emph{Hamiltonian reshaping} \cite{HuangTongFangSu2023learning,ma2024learningkbodyhamiltonianscompressed}, which inserts random single-qubit Pauli gates between evolutions of $H$. Each $\mu_s$ can be learned with Heisenberg-limited scaling using \emph{robust frequency estimation} \cite{KimmelLowYoder2015robust,ma2024learningkbodyhamiltonianscompressed}.

\subsubsection{Hamiltonian reshaping}

The goal of Hamiltonian reshaping is to approximate the time evolution of a specific term or a subset of commuting Pauli terms in $H$. 
We focus on the case where the target Hamiltonian after reshaping is a single traceless Pauli operator $P_s$ for illustration. Define $\mathcal{K}_{P_s}$ as the set containing all Pauli operators that commute with $P_s$. For any traceless Pauli operator $P_l\neq P_s$, the ensemble average of the transformed operator $P^{\dagger}P_lP$ over set $\mathcal{K}_{P_s}$ will result in zero, since half of the Pauli operators in $\mathcal{K}_{P_s}$ commute with $P_l$ and the other half anti-commute with it. On the other hand, for $P_s$, since every Pauli operator in $\mathcal{K}_{P_s}$ commutes with it, this transformation will leave $P_s$ unchanged. We implement this reshaping strategy by interleaving short-time evolutions under $H$, denoted $e^{-iH\tau}$, with random conjugations by Pauli operators $Q_k \in \mathcal{K}_{P_s}$. The resulting circuit takes the form:
\begin{equation}
    \label{eq:hamiltonian_reshaping}
    Q_{r_2} e^{-iH\tau}Q_{r_2}\cdots Q_2 e^{-iH\tau}Q_2Q_1 e^{-iH\tau}Q_1.
\end{equation}
Averaging over all possible choices of $Q_k$, this sequence implements a quantum channel that approximates evolution under an effective Hamiltonian which is $P_s$. This can be seen by analyzing a single time step quantum channel:
\begin{equation}
\begin{aligned}
    \rho &\mapsto \frac{1}{2^{2n-1}}\sum_{Q\in \mathcal{K}_{P_s}} Qe^{-iH\tau}Q\rho Qe^{iH\tau}Q 
    \\
    &= \rho - i\tau [H_{\mathrm{eff}},\rho]+\mathcal{O}(\tau^2) 
    ,
\end{aligned}
\end{equation}
where the effective Hamiltonian is
\begin{equation}
    H_{\mathrm{eff}} = \frac{1}{2^{2n-1}}\sum_{Q\in \mathcal{K}_{P_s}} QHQ=P_s.
\end{equation}

In the context of Hamiltonian learning, this procedure allows us to isolate the contribution of a single Pauli term $P_s$, as identified during the structure-learning phase. By applying this randomized reshaping channel over $r_2$ short-time steps, we approximate the desired evolution $e^{-i \mu_s P_s t}$, enabling accurate estimation of the coefficient $\mu_s$. This approach is closely related to randomized Hamiltonian simulation techniques such as qDRIFT~\cite{qDrift,ConcentrationForRPF}, and achieves diamond-norm error scaling of $\mathcal{O}(M^2t^2/r_2)$, where $M$ is the total number of Pauli terms in $H$ (see \Cref{sec:app_Hamiltonian_reshaping} for details).

\begin{alg}[htbp]{Coefficient learning subroutine $\mathcal{A}^{II}$}
\label{alg:algorithm_II}
\textbf{Input}: \vspace{-2mm}
\begin{enumerate}
    \item Unknown Hamiltonian $H$ with black-box access to its evolution $U_{H}(t)=e^{-iHt}$\vspace{-2mm}
    \item Desired Pauli operator $P_s$\vspace{-2mm}
    \item Desired learning accurarcy $\epsilon$ for cofficient $\mu_s$\vspace{-2mm}
    \item An estimated number of terms in the Hamiltonian $M$
\end{enumerate}\vspace{-2mm}

\textbf{Output}: Estimated coefficient $\hat{\mu}_s$ which satisfies $|\hat{\mu}_s-\mu_s|\leq \epsilon$\vspace{2mm}

\textbf{Pseudo-code}:

Set $a=-\pi$ and $b=\pi$

\textbf{For} $l = 1$ \textbf{to}
 $\log_{3/2}(2\pi/\epsilon)$ \textbf{do}:\vspace{-2mm}
    \begin{enumerate}
        \item Prepare $|\psi_{+}\rangle$ and $|\psi_{-}\rangle$ state in parallel. \\(see \Cref{sec:experimental_setup_of_A_2} for definitions of $|\psi_{+}\rangle,|\psi_{-}\rangle,O_+,O_-$).\vspace{-0mm}
        \item Set evolution time $\tau \gets \frac{\pi}{2(b-a)}$
        \item Set Trotterization steps $r_2\gets \Theta(M^2\tau^2)$
        \item Apply evolution with Hamiltonian reshaping: \\
        $|\psi_{+}'\rangle \gets \prod_{i=1}^{r_2}(Q_ie^{-iH\tau/r_2}Q_i)\ket{\psi_+}$
    
        $|\psi_{-}'\rangle \gets \prod_{i=1}^{r_2}(Q_ie^{-iH\tau/r_2}Q_i)\ket{\psi_-}$
        
        where $Q_i\in \mathcal{K}_{P_s}$ are sampled uniformly.
        \item Measure $O_+$ on $|\psi_{+}'\rangle$ and $O_-$ on $|\psi_{-}'\rangle$ each $\mathfrak{M}$ times and forms an estimate $\hat{S}=\langle O_{+}\rangle_{\psi_+'}+i\langle O_-\rangle_{\psi_-'}$ of $\exp(-i \frac{2\mu_s\pi}{b-a})$.
        \item If $\text{Im}\left(e^{-i\frac{(a+b)\pi}{2(b-a)}}\hat{S}\right)\leq 0$, then $b\gets (a+2b)/3$; else $a\gets (2a+b)/3$
    \end{enumerate}
Return $\hat{\mu}_s=(b-a)/4$
\end{alg}

\subsubsection{Robust frequency estimation}

With a good approximation of the time evolution channel under a single-term Hamiltonian $H_{\mathrm{eff}} = \mu_sP_s$ by Hamiltonian reshaping, we can robustly estimate $\mu_s$ to an accuracy $\epsilon$ with Heisenberg-limited scaling using the robust frequency estimation protocol. This protocol follows the same idea as the robust phase estimation protocol \cite{KimmelLowYoder2015robust}, but is modified to use only product-state input and to a good confidence interval rather than minimizing the mean-squared error. 

Consider a simple example with $P_s=IZZX$ and unknown $\mu_s$. We design the following two experiments: 1. the ($+$) experiment with input state $|\psi_{+}\rangle= \ket{0}\otimes\ket{+}\otimes \ket{0}\otimes\ket{+}$ and measuring observable $O_{+}=IXZX$ after evolution time $t$; 2. the ($-$) experiment with input state $|\psi_{-}\rangle=\ket{0}\otimes\ket{-}\otimes \ket{0}\otimes\ket{+}$ and measuring observable $O_{-}=IYZX$ after evolution time $t$. It can be easily verified that the expectation values are $\langle O_{+}\rangle=\cos(2\mu_s t)$ and $\langle O_{-}\rangle=\sin(2\mu_s t)$, which together form a simple oscillation $e^{i2\mu_st}$. The idea of robust frequency estimation is to narrow down the frequency $\theta$ of an oscillation signal $e^{i\theta t}$ in a hierarchical manner. Given $\theta\in [a,b]$, one can correctly distinguish whether $\theta$ is in $\theta\in[a,(a+2b)/3]$ or in $\theta\in [(2a+b)/3,b]$ by finite sampling in $t=\pi/(b-a)$. This process is then repeated $\mathcal{O}\left(\log_{3/2}((b-a)/\epsilon)
\right)$ times with the range divided by $2/3$ each time. 
The choice of the overlapped region makes this method robust against errors. 
In \Cref{sec:robust_frequency_estimation_details}, we describe the setup of the robust frequency estimation for arbitrary Pauli operator and show it is robust against finite-sample error and channel approximation error from the Hamiltonian reshaping. 

We provide a high-level pseudocode in Algorithm~\hyperref[alg:algorithm_II]{2}. The inputs to this subroutine are: (1) black-box access to the time evolution generated by the unknown Hamiltonian H; (2) the desired Pauli term $P_s$ for coefficient estimation; (3) the desired learning accuracy $\epsilon$ for coefficient $\mu_s$; and (4) the number of terms in the Hamiltonian $M$.

\subsection{Hierarchical learning\label{subsec:hierarchy_learning}}
As we have described above, $\mathcal{A}^I$ (or $\mathcal{A}^{I'}$) can learn the structure of the Hamiltonian and $\mathcal{A}^{II}$ can estimate the coefficients. However, directly concatenating $\mathcal{A}^I$ (or $\mathcal{A}^{I'}$) with $\mathcal{A}^{II}$ will give an algorithm with a complexity of $\widetilde{\mathcal{O}}(1/\epsilon^3)$ dependence on the target accuracy $\epsilon$, not enough for achieving Heisenberg-limited scaling. 
We address this issue by introducing a hierarchical learning protocol, where we divide the terms in $H$ into $J=\left\lceil\log_2(1/\epsilon)\right\rceil$ levels. For the $j$-th level, we learn all terms with coefficient $2^{-(j+1)}<|\mu_s|\leq 2^{-j}$ to $\epsilon$ accuracy by applying $\mathcal{A}^I$ (or $\mathcal{A}^{I'}$) and $\mathcal{A}^{II}$. Crucially, since we already have an estimation of coefficients for terms with $2^{-j}\leq\mu_s$, we can approximately cancel out those terms by interfacing the black-box evolution with the quantum computer in Hamiltonian simulation. 

More concretely, let $\hat{H}_j$ denote the partial Hamiltonian composed of terms with $|\mu_s| > 2^{-j}$ that have been learned in previous rounds. Using this classical description, we can approximate the evolution under the residual Hamiltonian $H - \hat{H}_j$ by interleaving the black-box evolution $e^{-iH\tau}$ with the engineered evolution $e^{i\hat{H}_j\tau}$ implemented on a quantum computer. Specifically, we construct the sequence $(e^{-iH t/r} e^{i\hat{H}_j t/r})^r$, which approximates $e^{-i(H - \hat{H}_j)t}$ via the Lie product formula. This procedure effectively isolates the remaining unlearned terms with $|\mu_s| \leq 2^{-j}$, without requiring access to the time-reversed evolution $e^{+iHt}$.

\begin{alg}[htbp]{Ansatz-free Hamiltonian learning algorithm}
\label{alg:algorithm_all}
    \textbf{Input}: \vspace{-2mm}
\begin{enumerate}
    \item Unknown Hamiltonian $H$ with black-box access to its evolution $U_{H}(t)=e^{-iHt}$\vspace{-2mm}
    \item Desired learning accurarcy $\epsilon$ for cofficients\vspace{-2mm}
    \item An estimated number of terms in the Hamiltonian $M$\vspace{-2mm}
    \item Constant number of measurements $\mathfrak{M}$
\end{enumerate}\vspace{-2mm}

\textbf{Output}: Estimated coefficient $\hat{\mu}_s$ which satisfies $|\hat{\mu}_s-\mu_s|\leq \epsilon$\vspace{2mm}

\textbf{Pseudo-code}:\vspace{0mm}

Set $\hat{H}=0$,

\textbf{For} $j = 0$ \textbf{to} $\lceil\log_{2}(1/\epsilon)\rceil-1$ \textbf{do}:\vspace{-2mm}
\begin{enumerate}
    \item $\mu_m = 2^{-(j+1)}$ \vspace{-2mm}
    \item $\mathcal{S}_j=\mathcal{A}^{I}((H,U_H(t)),\hat{H},\mu_m,M,\mathfrak{M})$ \vspace{-2mm}
    \item \textbf{For} $s\in \mathcal{S}_j$ \textbf{do}:
    
        $\hat{\mu}_s=\mathcal{A}^{II}((H,U_H(t)),P_s,\epsilon,M)$
    \vspace{-2mm}
    \item Update $\hat{H}\gets \hat{H}+\sum_{s\in \mathcal{S}_j}\hat{\mu}_sP_s$
\end{enumerate}\vspace{-2mm}
Return $\hat{H}$
\end{alg}

With the approximated time evolution channel with all learned large terms cancelled out, we then rescale the Hamiltonian by extending the evolution time from $t$ to $t/2^{-j}$, thereby boosting the initially small coefficients in $j$-th level to a constant scale. The steps of this procedure are depicted in \Cref{fig:protocol} (c). Both the term cancellation and rescaling can be achieved using Trotterization (see \Cref{app:structure_learning} for detailed error analysis). We also provide a high-level pseudocode in Algorithm~\hyperref[alg:algorithm_all]{3}.

\subsection{Total time complexity analysis\label{subsec:total_time_complexity}}

We now analyze the scaling of the total experimental time. As defined in \Cref{eq:total_time_main}, this quantity is given by the sum of all evolution times multiplied by the number of experimental repetitions. At each level $j$, we first perform structure learning to identify the support set
\eqs{
\mathfrak{S}_j=\{s:2^{-(j+1)}<|\mu_s|\leq 2^{-j}\}
}
In each such experiment, the evolution time scales as $t_{j,l}^{(1)}=\mathcal{O}(2^j/M)$, where the $2^j$ factor arises from rescaling the Hamiltonian to boost small coefficients. As discussed previously, we require $L=\widetilde{\mathcal{O}}(M^2)$ repetitions to ensure that each $s\in \mathfrak{S}_j$ is sampled \emph{at least once} with high probability. Thus, the total experimental time spent on structure learning at level $j$ is
\eqs{
T_{1}^{j\text{th}}=Lt_{j,l}^{(1)}=\widetilde{\mathcal{O}}(2^{j}M).
}

Next, we perform coefficient learning using robust frequency estimation to estimate each $\mu_s$ for $s\in \mathfrak{S}_j$ to within precision $\epsilon$. Each experiment requires evolution time $t^{(2)}_{j,l}=\widetilde{O}(1/\epsilon)$, and the number of repetitions per term is constant due to the robustness of the frequency estimation method, which tolerates a constant additive error. In the worst case, we need to estimate up to $L_j=\mathcal{O}(M^2)$ coefficients at level $j$, so the total time for coefficient learning at level $j$ is
\eqs{
T_2^{j\text{th}}=L_jt^{(2)}_{j,l}=\widetilde{\mathcal{O}}(M^2/\epsilon).
}

Combining the contributions from structure and coefficient learning across all levels $j\in\{0,1,\ldots,\left\lceil\log_2(1/\epsilon)\right\rceil-1\}$, the total experimental time is:
\eqs{
T&=\sum_{j=0}^{\lceil \log_{2}(1/\epsilon)\rceil-1}(T_1^{j\text{th}}+T_2^{j\text{th}})\\
&=\sum_{j=0}^{\lceil \log_{2}(1/\epsilon)\rceil-1}(\widetilde{\mathcal{O}}(2^j M)+\widetilde{\mathcal{O}}(M^2/\epsilon))\\
&=\widetilde{\mathcal{O}}(M/\epsilon)+\widetilde{\mathcal{O}}(M^{2}/\epsilon),
}
which confirms that the overall protocol achieves Heisenberg-limited scaling, $T = \widetilde{\mathcal{O}}(1/\epsilon)$, in total experimental time. In addition, we show this protocol is robust against SPAM errors (see \Cref{app:structure_learning}) and the time complexity with the alternative approach is provided in \Cref{sec:single-copy}.

\section{Applications}

To demonstrate the effectiveness of our method, we consider three representative Hamiltonian learning tasks and benchmark our performance against the state-of-the-art Heisenberg-limited learning approach introduced in Ref.~\cite{HuangTongFangSu2023learning}. Our ansatz-free method achieves the same gold-standard Heisenberg scaling—where the total experimental time $T_{\text{total}}$ scales inversely with the target precision $\epsilon$ as $T_{\text{total}} \sim \epsilon^{-1}$—while offering significantly broader applicability. In particular, our method accurately recovers \emph{non-local} and \emph{multi-qubit} interactions without relying on a predefined ansatz. This is especially important in scenarios where the structure of the effective Hamiltonian is unknown or difficult to model, making ansatz-free learning a critical tool for validating analog pulse design and analog quantum simulation protocols.

Our first example involves a disordered XY Hamiltonian with nearest-neighbor interactions, augmented by additional non-local couplings (mimicking cross-talk) and an all-to-all interaction term. We applied both our ansatz-free method and the state-of-the-art Heisenberg-limited learning approach \cite{HuangTongFangSu2023learning} (referred to as the vanilla method) to learn the Hamiltonian coefficients. The results are shown in \figref{fig:example1}(a) as a bar plot. Due to the geometric locality ansatz, the vanilla method fails to capture the non-local and many-body terms, effectively zeroing them out during its divide-and-conquer procedure. This is highlighted by the missing yellow bars within the red dashed box. In contrast, our ansatz-free method accurately recovers all significant interactions, including the non-local and multi-qubit terms, with just one iteration of structure and coefficient learning. The inset of \figref{fig:example1}(a) shows the total experimental time versus target precision $\epsilon$ for learning one Pauli operator. An empirical fit yields $T_{\text{total}} \sim \epsilon^{-0.96}$, in close agreement with the theoretically optimal Heisenberg-limited scaling as we proved in the previous sections.

\begin{figure}[!htbp]
    \centering
    \includegraphics[width=0.99\linewidth]{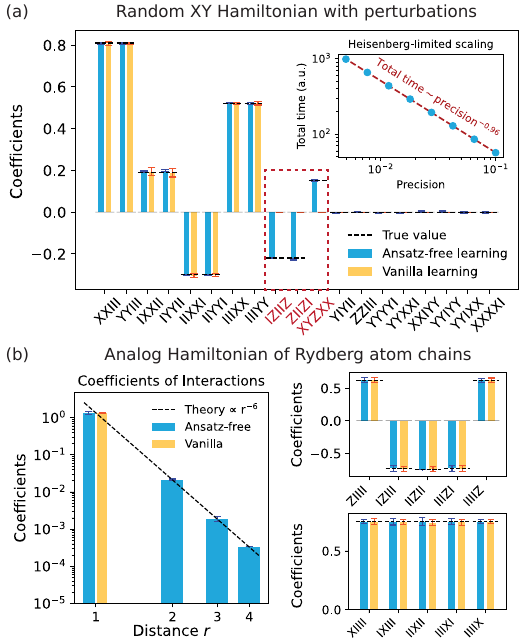}
    \caption{\textbf{Ansatz-free Hamiltonian learning with non-local and many-body interactions.} (a)  Learning a 1D disordered XY model with nearest-neighbor interactions ($X_iX_{i+1} + Y_iY_{i+1}$), supplemented by additional non-local cross-talk and an all-to-all coupling. The ansatz-free method (this work) is compared with the vanilla approach based on a nearest-neighbor ansatz \cite{HuangTongFangSu2023learning}. Each run uses 2000 measurements for structure learning and 1000 for coefficient estimation. The inset shows total measurement time versus target precision; a fit yields $T_{\text{total}} \propto \text{precision}^{-0.96}$, consistent with Heisenberg-limited scaling. The red dashed box highlights failure of the vanilla method to recover non-local and many-body terms. False positives in structure learning (Pauli strings) are shown; coefficient learning confirms them as negligible.
(b) Learning the Hamiltonian of a 1D neutral atom chain with $10\mu \text{m}$ spacing, uniform Rabi drive ($\Omega = 1.5/2\pi$ MHz), and global detuning ($\Delta = -4/2\pi$ MHz). Our method identifies the correct power-law decay from dipolar interactions after 5 rounds of structure and coefficient learning, using the same measurement counts as in (a). In contrast, the vanilla method captures only nearest-neighbor terms. All error bars indicate one standard deviation.}
    \label{fig:example1}
\end{figure}

 In many natural quantum systems, interactions often exhibit long-range tails, and uncovering the full interaction profile — free from ansatz-induced bias — is essential for revealing the underlying physics. A prominent example arises in neutral atoms or polar molecules trapped in optical tweezers, where atoms in highly excited Rydberg states interact via long-range forces. If one encodes a qubit in the ground and the highly excited states also known as Rydberg states, the dominant interaction between nearby atoms resembles an Ising-type $Z_iZ_j$ coupling. However, the interaction is not strictly local — it decays algebraically with distance, and the precise power-law exponent depends on the physical origin, such as van der Waals ($1/r^6$) or dipole-dipole ($1/r^3$) interactions. Capturing these subtle but physically significant long-range features requires an ansatz-free Hamiltonian learning protocol that does not assume a  structure of the interactions. Moreover, to resolve small-magnitude couplings efficiently, such a method must attain Heisenberg-limited scaling, ensuring maximal learning efficiency and minimal experimental time within the fundamental limits set by quantum mechanics. 

 To illustrate this, we simulate a physical Hamiltonian for a chain of five neutral atoms driven by a uniform Rabi frequency $\Omega$ and a global detuning $\Delta$. The system is described by the Hamiltonian
\eqs{
H/\hbar = &\frac{\Omega}{2} \sum_{l} \left( |g_l\rangle \langle r_l| + |r_l\rangle \langle g_l| \right) - \Delta \sum_{l} |r_l\rangle \langle r_l| \\
&+ \sum_{j<l} V_{jl} |r_j\rangle \langle r_j| \otimes |r_l\rangle \langle r_l| ,
}
where $V_{jl} = \mathrm{C} / |\vec{x}_j - \vec{x}_l|^6$ describes the van der Waals interaction between atoms $j$ and $l$. We use $\mathrm{C} = 862{,}690 \times 2\pi$ MHz·$\mu$m$^6$, corresponding to the interaction strength for Rydberg states $70S_{1/2}$ of $^{87}\text{Rb}$ atoms. The atoms are equally spaced with a separation of $10\mu\text{m}$, and the drive parameters are set to $\Omega = 1.5/2\pi$ MHz and $\Delta = -4/2\pi$ MHz. Due to the $1/r^6$ decay of the van der Waals interaction, the interaction strength rapidly diminishes with distance.

As shown in \figref{fig:example1}(b), the vanilla learning method successfully captures the dominant nearest-neighbor interactions (yellow bars) but fails to identify the weaker long-range components, as evident in the left panel. In contrast, our ansatz-free method accurately reconstructs all significant terms, including the long-range tails, after five iterations of structure and coefficient learning. Remarkably, the recovered interaction strengths closely follow the theoretical $1/r^6$ decay expected from van der Waals interactions. A key advantage of our Heisenberg-limited learning framework is that the number of measurements per iteration remains fixed at 2000 for structure learning and 1000 for coefficient estimation, regardless of the interaction strength. This allows the method to efficiently resolve weak interactions that are several orders of magnitude smaller than the leading terms. In contrast, learning algorithms that do not achieve Heisenberg-limited scaling must require an increasing number of measurements to resolve smaller couplings, making them substantially less efficient for identifying subtle features in the Hamiltonian. In Appendix~\ref{sec:numerics_details}, we provide additional details of the numerical simulations and include plots of all Pauli terms identified at each iteration, including false positives.

\begin{figure}[!htbp]
    \centering
    \includegraphics[width=0.99\linewidth]{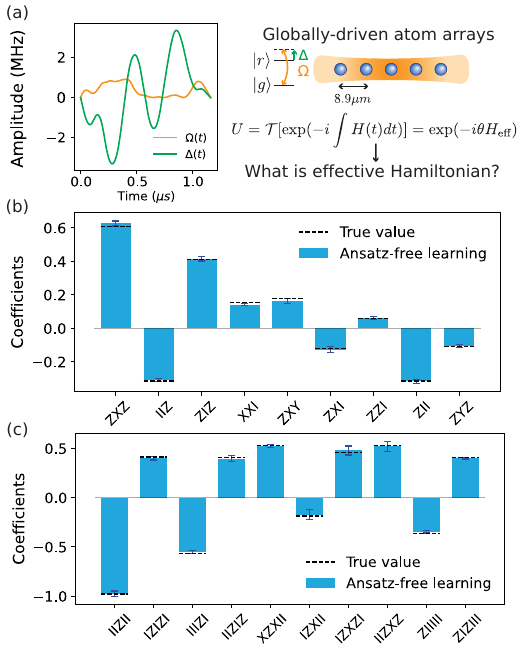}
    \caption{\textbf{Verifying effective Hamiltonians in engineered analog quantum simulations.} (a) A 1D array of three Rydberg atoms spaced by $8.9~\mu\text{m}$ is driven by time-dependent Rabi frequency $\Omega(t)$ and global detuning $\Delta(t)$, engineered to approximate evolution under a symmetry-protected topological (SPT) Hamiltonian with three-body terms $Z_iX_{i+1}Z_{i+2}$, achieving a unitary fidelity of 90$\%$. Verifying the resulting effective Hamiltonian is essential for validating analog quantum simulations and pulse engineering. Applying our ansatz-free learning method, we identify ZXZ as the leading interaction. (b) Extending the same pulse to a five-atom array results in a different effective Hamiltonian. Our method successfully identifies all significant Pauli terms, revealing that the three-atom ansatz fails to generalize. True values are obtained via exact simulation; error bars represent one standard deviation.}
    \label{fig:example2}
\end{figure}

Ansatz-free Hamiltonian learning also serves as a powerful tool for validating analog quantum simulations, acting as an unbiased verifier of the effective dynamics \cite{QuantumVerification,EntanglementHamiltonianLearning}. A central challenge—and opportunity—in analog quantum simulation lies in designing robust pulse sequences that realize effective Hamiltonians beyond the hardware’s native interactions. This task is closely connected to multi-qubit gate synthesis and analog gadget engineering. To illustrate this application, we consider a time-dependent Hamiltonian implemented in a one-dimensional Rydberg atom array, as used in a recent experiment \cite{Hu2025expressivity}, where the control fields $\Omega(t)$ and $\Delta(t)$ are shown in \figref{fig:example2}(a). This pulse sequence is engineered to approximate the effective time evolution of the cluster-Ising model with three-body ZXZ interactions:
\eqs{
U = \mathcal{T} \exp\left(-i \int H(t), dt\right) = \exp(-i\theta H_{\text{eff}}),
}
where $H_{\text{eff}} = \sum_i Z_i X_{i+1} Z_{i+2}$ and $\theta = 0.1$. Due to hardware limitations, the resulting evolution achieves a unitary fidelity of approximately $91\%$ in the three-atom system. This motivates the need to independently learn and verify the effective Hamiltonian realized by the pulse.

Using our ansatz-free method, we performed Hamiltonian learning on both three-atom and five-atom chains driven by the same global time-dependent controls, assuming only that the target rotation angle is $\theta \approx 0.1$. As in previous examples, we used 2000 measurements for structure learning and 1000 for coefficient estimation. The results, shown in \figref{fig:example2}(b,c), confirm that all relevant terms in the effective Hamiltonian are accurately recovered. For the three-atom system, the ZXZ interaction emerges as the dominant term, accompanied by smaller perturbative corrections. However, when the same control sequence is applied to a five-atom system, the resulting effective Hamiltonian differs, despite still featuring global ZXZ components. Additional details regarding the numerical simulations are provided in \Cref{sec:numerics_details}.

This observation highlights a crucial insight: ansatz derived from small systems (e.g., three atoms) may not generalize to larger systems. Imposing such ansatz can introduce significant bias and lead to the omission of important interaction terms. As demonstrated in \figref{fig:example1}(a), our ansatz-free approach avoids this pitfall and enables unbiased reconstruction of complex, engineered Hamiltonians.

\section{Trade-off between total evolution time and quantum control}

\begin{figure*}[ht]
    \centering
    \includegraphics[width=0.99\textwidth]{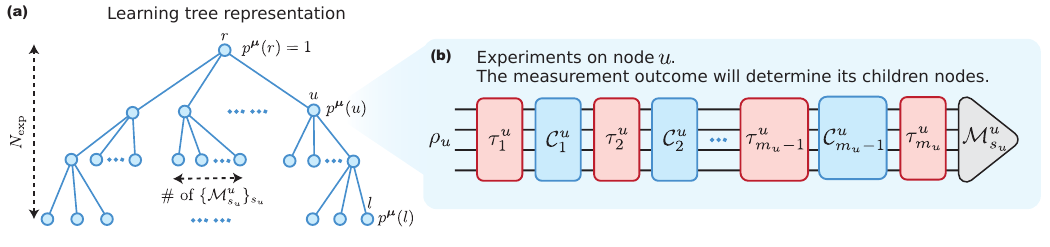}
    \caption{(a) Learning tree representation. Each node $u$ represents an experiment. Starting from the root experiment $r$, the number of child nodes depends on the possible POVM measurements. The transition probabilities are determined by Born's rule (also see \Cref{app_eq:transition_prob}). After $N_\text{exp}$ experiments, one arrives at the leaves $l$ of the learning tree. (b) In each node u, the learning model prepares an arbitrary state $\rho_u$, applies discrete control channels $\mathcal{C}^u_{k}$, and queries real-time Hamiltonian evolutions with time $\tau_k^{u}$ multiple times. The protocols can also incorporate ancilla qubits (quantum memory).
 }
    \label{fig:main_learn_lower_model}
\end{figure*}

Compared to previous approaches that learn Hamiltonians from eigenstates or Gibbs states, our protocol surpasses the standard quantum limit and achieves Heisenberg-limited scaling. This is accomplished by sequentially querying the real-time dynamics of the system, interleaved with discrete quantum controls implemented via Trotterized quantum evolution on a quantum computer. A natural question arises: are quantum controls or having a quantum computer essential for achieving Heisenberg-limited scaling? We answer this affirmatively by proving a trade-off between the total evolution time required to learn an unknown Hamiltonian $H$ and the number of discrete control operations. We emphasize that here, ``control" refers specifically to unitary operations we apply during time evolution, not controlled-unitary gates.

Intuitively, we can generalize experiments with discrete quantum controls and sequential queries to the Hamiltonian as the following theoretical model. 
Given a Hamiltonian $H$ with coefficient vector $\bm{\mu}$ defined in \eqref{eq:unknown_Hamiltonian}, the protocol performs multiple experiments and measures at the end of each experiment.
In each experiment, the protocol prepares an input state (possibly with ancilla qubits) and queries the Hamiltonian multiple times with a discrete quantum control channel between every two neighboring queries, as illustrated in \figref{fig:main_learn_lower_model}(b).

The entire process, including all experimental branches and outcomes, can be organized into a tree structure—referred to as the learning tree representation, in which each node $u$ (see \figref{fig:main_learn_lower_model}(a)) corresponds to a particular experiment. The protocol can be adaptive in the sense that it can dynamically decide how to prepare the input state, query the real-time evolutions, perform quantum controls, and measure the final state based on the history of the previous experiments. Each node in the tree is associated with a probability $p^{\bm{\mu}}(u)$, with the root node having $p^{\bm{\mu}}(r) = 1$. The probability of reaching any other node is determined by the full measurement trajectory leading from the root, in accordance with Born’s rule. A formal definition of the learning tree model is provided in \Cref{app:lower}.

In a recent work~\cite{dutkiewicz2024advantage}, it is shown through quantum Fisher information~\cite{PhysRevD.23.357,PhysRevLett.72.3439,BRAUNSTEIN1996135} that for \emph{unbiased} learning algorithm, Heisenberg-limited scaling is not possible in the absence of quantum control. 
In contrast, we establish a rigorous and fully general lower bound that applies to any adaptive protocol—whether biased or unbiased—that may employ ancillary quantum memory. Our result leverages tools from quantum Fisher information theory and the learning tree representation to characterize the fundamental trade-offs in Hamiltonian learning.
\begin{result}[Informal version of \Cref{thm:lower}]\label{thm:lower_informal}
For any protocol that can be possibly adaptive, biased, and ancilla-assisted, with total evolution time $T$ and at most $\mathcal{L}$ discrete quantum controls per experiment, there exists some Hamiltonian $H$ that requires $T = \Omega(\mathcal{L}^{-1}\epsilon^{-2})$ to estimate the coefficient of the unknown Hamiltonian $H$ within additive error $\epsilon$.

\end{result}
To establish this result, we consider a distinguishing task between two Hamiltonians: the original Hamiltonian $H(\bm{\mu})$ and a perturbed version $H(\bm{\mu} + \delta\bm{\mu})$, where the perturbation satisfies $|\delta\bm{\mu}|_{\infty} \propto \epsilon$. The objective is to determine which of the two Hamiltonians governs the underlying dynamics. It is straightforward to see that any Hamiltonian learning algorithm capable of estimating $\bm{\mu}$ can also be used to solve this distinguishing problem.

Within the learning tree framework, a necessary condition for successful discrimination between these two hypotheses—according to Le Cam’s two-point method~\cite{yu1997assouad}—is that the total variation distance between the probability distributions over the leaf nodes of the two learning trees must be large. Quantitatively, this condition requires
\eqs{
\frac{1}{2} \sum_{\ell \in \text{leaf}(\mathcal{T})} \left| p^{\bm{\mu} + \delta\bm{\mu}}(\ell) - p^{\bm{\mu}}(\ell) \right| = \Theta(1). \label{eq:leaf_prob_diff}
}

Using the martingale likelihood ratio technique developed in a series of recent works~\cite{ComplexityNISQ,2023arXiv230914326C,2021arXiv211207646C}, the total variation lower bound can be translated into a quantitative statement via Lemma 7 of Ref.~\cite{OptimalTradeOff}. Specifically, if there exists a constant $\Delta > 0$ such that
\begin{align}\label{eq:martingale}
\mathbb{E}_{s_u\sim p^{\bm{\mu}}(s_u|u)}[(L_{\delta\bm{\mu}}(u,s_u)-1)^2]\leq\Delta,
\end{align}
where
\begin{align}
L_{\delta\bm{\mu}}(u, s_u) = \frac{p^{\bm{\mu} + \delta\bm{\mu}}(s_u \mid u)}{p^{\bm{\mu}}(s_u \mid u)},
\end{align}
then the number of experiments must satisfy $N_{\exp}\geq\Omega(1/\Delta)$. The physical intuition behind this result is as follows: when the distributions  $p^{\bm{\mu}}(s_u|u)$ and $p^{\bm{\mu}+\delta \bm{\mu}}(s_u|u)$ are nearly indistinguishable at all internal nodes of the learning tree, a deep tree, corresponding to a large number of experimental queries, is required to amplify these differences to a detectable level at the leaves, as captured in \Cref{eq:leaf_prob_diff}. In Appendix~\ref{app:lower}, we show that
\eqs{
&\mathbb{E}_{s_u\sim p^{\bm{\mu}}(s_u|u)}[(L_{\delta\bm{\mu}}(u,s_u)-1)^2] \\
&\leq \max_{\bm{v}:\norm{\bm{v}}_2 3\sqrt{M}\epsilon}\bm{v}^\top I^{(Q)}_{{\bm{\mu}},u}\bm{v},
}
where $I^{(Q)}_{{\bm{\mu}},u}$ is the quantum Fisher information matrix associated with node $u$. Combining this result with the Lemma 17 from Ref.~\cite{dutkiewicz2024advantage}, we establish the claim stated in \Cref{thm:lower_informal}.

This result implies that, in general, no protocol can achieve Heisenberg-limited scaling without employing at least $\mathcal{L} = \Omega(1/\epsilon)$ discrete quantum controls per experiment. This highlights a fundamental trade-off between quantum controllability and learning efficiency. However, this trade-off should not be misinterpreted as allowing Hamiltonian learning protocols to surpass the Heisenberg limit by simply increasing the number of discrete controls. In fact, as shown in Refs.~\cite{HuangTongFangSu2023learning,ma2024learningkbodyhamiltonianscompressed,dutkiewicz2024advantage}, there exists a strict lower bound on the total evolution time, namely $T = \Omega(1/\epsilon)$, that applies regardless of the number of available controls.

\section{Outlook}
In this work, we have introduced the \emph{first} ansatz-free Hamiltonian learning protocol for an $n$-qubit system that makes no assumptions about the interaction structure. Our protocol achieves the gold standard of Heisenberg-limited scaling, with total experimental time scaling as $T=\mathcal{O}(1/\epsilon)$, and remains polynomial in the number of relevant interactions. Importantly, the protocol relies solely on black-box access to the time evolution generated by the unknown Hamiltonian, along with elementary digital controls on a quantum computer. We demonstrate numerically that our method can efficiently learn non-local and many-body interactions in disordered spin chains and Rydberg atom arrays. Furthermore, we establish a fundamental trade-off between total evolution time and quantum control in general Hamiltonian learning protocols. This result offers an intriguing bridge between quantum learning theory and quantum metrology and may serve as a foundation for future theoretical exploration.

Our work opens several promising directions for future research. First, ansatz-free learning provides a critical tool for the verification of analog quantum simulators. In quantum many-body physics, objects such as the entanglement Hamiltonian and the entanglement negativity Hamiltonian play a key role in characterizing quantum correlations in both closed and open systems \cite{EntanglementHamiltonianLearning,Cardy_2016,EntanglementHamiltonianReview,NegativityHamiltonian,2025arXiv250609561T}. While these Hamiltonians are analytically known in only certain solvable models, their structure remains largely unexplored in generic interacting systems. Ansatz-free Hamiltonian learning offers a new tool to empirically probe and verify these objects on programmable quantum devices. In particular, the entanglement negativity Hamiltonian is generally non-Hermitian, motivating the development of new learning protocols tailored to non-Hermitian Hamiltonians \cite{PhysRevB.110.235113}.

Another important direction is to understand the impact of noise during black-box Hamiltonian queries on learning efficiency. Preliminary analysis suggests the existence of a threshold for learning precision $\epsilon$, such that Heisenberg-limited scaling persists when $\epsilon> \mathcal{O}(\eta)$, where $\eta$ characterizes the strength of the noise. Below this threshold, the learning performance degrades, and Heisenberg scaling breaks down. Investigating this transition is an exciting avenue for future work. In particular, techniques from quantum metrology may offer valuable insights in this direction \cite{MetrologyQEC,EnhancedSensing}. Finally, extending ansatz-free learning to the open-system setting is another promising direction. Developing ansatz-free Lindbladian learning protocols would have broad applications in quantum device benchmarking, quantum noise characterization, and open-system dynamics more generally \cite{2024arXiv240803376S,PhysRevA.105.032435}.

\section{Acknowledgements} 
We are thankful for the insightful discussions with Sitan Chen, Friedrich Liyuan Chen, Soonwon Choi, Dong-Ling Deng, Nik O. Gjonbalaj, Yingfei Gu, Hsin-Yuan Huang, Christian Kokail, Yunchao Liu, Francisco Machado, Daniel K. Mark, and Pengfei Zhang. HYH and SFY acknowledge the support from DOE through the QUACQ
program (DE-SC0025572). Author contributions: HYH, MM, and QY conceived the project during Ma’s and Ye’s visit to Harvard. HYH and MM carried out the majority of the calculations. WG and YT contributed to the theoretical proofs. STF and SFY supervised the project throughout.

\newpage

\let\oldaddcontentsline\addcontentsline
\renewcommand{\addcontentsline}[3]{}
\bibliography{ref.bib}
\let\addcontentsline\oldaddcontentsline
\clearpage
\newpage

\onecolumngrid

\begin{appendix}
\clearpage

\renewcommand{\appendixname}{Appendix}
\newcommand{\appendixarabic}{
    \renewcommand{\thesubsection}{\arabic{subsection}}
    \renewcommand{\thesubsubsection}{\arabic{subsection}.\arabic{subsubsection}}
    \makeatletter
    \renewcommand{\p@subsection}{}
    \renewcommand{\p@subsubsection}{}
    \makeatother
}

\renewcommand{\figurename}{Supplementary Figure}
\setcounter{secnumdepth}{3}
\makeatletter
     \@addtoreset{figure}{section}
\makeatother

\setcounter{figure}{0}
\renewcommand{\figurename}{Fig.}
\renewcommand{\thefigure}{S\arabic{figure}}
\setcounter{table}{0}
\renewcommand{\tablename}{Table}
\renewcommand{\thetable}{S\arabic{table}}

\begin{center}
    \textbf{CONTENTS}
\end{center}

\smallskip

\noindent \ref{app:notations}.~~\hyperref[app:notations]{NOTATIONS}\dotfill\textbf{\pageref{app:notations}}
\medskip

\noindent \ref{app:main_ancillary}.~~\hyperref[app:main_ancillary]{MAIN RESULTS AND PROOF IDEAS}\dotfill\textbf{\pageref{app:main_ancillary}}
\medskip

\noindent 
\ref{app:structure_learning}.~~\hyperref[app:structure_learning]{HAMILTONIAN STRUCTURE LEARNING $\mathcal{A}^{I}_j$} \dotfill\textbf{\pageref{app:structure_learning}}
\medskip

\noindent \qquad \begin{minipage}{\dimexpr\textwidth-0.82cm}
  \hyperref[app:Trotter]{1. Hamiltonian simulation with Trotterization}\dotfill\textbf{\pageref{app:Trotter}}
  \newline
\hyperref[app:Taylor]{2. Taylor expansion of the time evolution operator}\dotfill\textbf{\pageref{app:Taylor}}
\newline
\hyperref[sec:Bell_sampling]{3. Bell sampling}\dotfill\textbf{\pageref{sec:Bell_sampling}}
\newline
\hyperref[sec:SPAM_error]{4. SPAM error}\dotfill\textbf{\pageref{sec:SPAM_error}}
\newline
\hyperref[sec:lower_bound_prob]{5. Lower bound for the probability of elements in $\mathfrak{S}_j$}\dotfill\textbf{\pageref{sec:lower_bound_prob}}
\newline
\hyperref[sec:high_prob_detection]{6. Determine the high-probability elements of a sparse probability distribution}\dotfill\textbf{\pageref{sec:high_prob_detection}}
\newline
\hyperref[sec:complexity_I]{7. The complexity of $\mathcal{A}_j^{I}$}\dotfill\textbf{\pageref{sec:complexity_I}}
 \end{minipage}
\medskip

\noindent 
\ref{app:coef_learning}.~~\hyperref[app:coef_learning]{HAMILTONIAN COEFFICIENT LEARNING $\mathcal{A}^{II}_j$} \dotfill\textbf{\pageref{app:coef_learning}}
\medskip

\noindent \qquad \begin{minipage}{\dimexpr\textwidth-0.82cm}
  \hyperref[sec:app_Hamiltonian_reshaping]{1. Hamiltonian reshaping}\dotfill\textbf{\pageref{sec:app_Hamiltonian_reshaping}}
  \newline
\hyperref[sec:robust_frequency_estimation_details]{2. Robust frequency estimation}\dotfill\textbf{\pageref{sec:robust_frequency_estimation_details}}
\newline
\hyperref[sec:experimental_setup_of_A_2]{3. The experimental setup of $\mathcal{A}_j^{II}$}\dotfill\textbf{\pageref{sec:experimental_setup_of_A_2}}
\newline
\hyperref[sec:complexities_of_A_2]{4. The complexity of $\mathcal{A}_j^{II}$}\dotfill\textbf{\pageref{sec:complexities_of_A_2}}
 \end{minipage}
\medskip

\noindent 
\ref{sec:total_time_complexity}.~~\hyperref[sec:total_time_complexity]{TOTAL EVOLUTION TIME COMPLEXITY} \dotfill\textbf{\pageref{sec:total_time_complexity}}
\medskip

\noindent 
\ref{sec:single-copy}.~~\hyperref[sec:single-copy]{ALTERNATIVE SINGLE-COPY PRODUCT STATE INPUT APPROACH $\mathcal{A}^{I'}_j$} \dotfill\textbf{\pageref{sec:single-copy}}
\medskip

\noindent \qquad \begin{minipage}{\dimexpr\textwidth-0.82cm}
  \hyperref[sec:Pauli_error_rate]{1. Pauli error rate for time-evolution channel}\dotfill\textbf{\pageref{sec:Pauli_error_rate}}
  \newline
\hyperref[sec:population_recovery]{2. Estimating Pauli error rates via population recovery}\dotfill\textbf{\pageref{sec:population_recovery}}
\newline
\hyperref[sec:complexity_single_copy]{3. Complexity of $\mathcal{A}^{I'}_j$}\dotfill\textbf{\pageref{sec:complexity_single_copy}}
\newline
\hyperref[sec:total_evo_time_for_single_copy]{4. Totoal evolution time complexity of the single-copy product state input Hamiltonian learning protocol}\dotfill\textbf{\pageref{sec:total_evo_time_for_single_copy}}

 \end{minipage}
\medskip

\noindent 
\ref{app:lower}.~~\hyperref[app:lower]{THE PROOF FOR THE TRADE-OFF} \dotfill\textbf{\pageref{app:lower}}
\medskip

\noindent 
\ref{sec:numerics_details}.~~\hyperref[sec:numerics_details]{DETAILS OF THE SIMULATION} \dotfill\textbf{\pageref{sec:numerics_details}}
\medskip

\section{Notations\label{app:notations}}
In this work, the Pauli matrices are denoted by $\sigma^{x},\sigma^{y},\sigma^{z}$.
We use the following notation to denote the Pauli eigenstates:
\begin{equation}
    \label{eq:pauli_lambdastates}
    \begin{aligned}
        &\ket{1,z} = \ket{0},\quad \ket{-1,z} = \ket{1},\quad \ket{1,x} = \ket{+},\quad \ket{-1,x} = \ket{-},\\
        &\ket{1,y} = \frac{1}{\sqrt{2}}(\ket{0}+i\ket{1}),\quad \ket{-1,y} = \frac{1}{\sqrt{2}}(\ket{0}-i\ket{1}).
    \end{aligned}
\end{equation}
We denote the set of all $N$-fold tensor products of single-qubit Pauli matrices (and the identity) by $\mathbb{P}_n$:
\begin{equation}
    \label{eq:defn_pauli_matrices}
    \mathbb{P}_n = \left\{\bigotimes_{i=1}^n P_i:P_i=I,\sigma^x,\sigma^y,\text{ or }\sigma^z\right\}.
\end{equation}

We use the following notation for the four maximally entangled 2-qubit states
\eqs{
&\ket{\Phi^{+}} = \dfrac{\ket{00}+\ket{11}}{\sqrt{2}}\\
&\ket{\Phi^{-}}=\dfrac{\ket{00}-\ket{11}}{\sqrt{2}}\\
&\ket{\Psi^{+}}=\dfrac{\ket{10}+\ket{01}}{\sqrt{2}}\\
&\ket{\Psi^{-}}=i\dfrac{\ket{10}-\ket{01}}{\sqrt{2}},
}
which are also the four eigenstates of the Bell-basis measurement. In the following, we will also use $\ket{\EPR_n}=\ket{\Phi^{+}}^{\otimes n}$.

We consider the $N$-qubit traceless Hamiltonian to be learned as
\begin{equation}
    H = \sum_s\mu_sP_s,\quad |\mu_s|\leq 1,
\end{equation}
where $P_s\in \mathbb{P}_N\setminus\{I^{\otimes N}\}$. 
Let $M$ be the number of terms in the coefficient vector $\boldsymbol{\mu}=(\mu_1,\mu_2,\ldots,\mu_{4^N-1})^T$ which are larger than the threshold constant $\epsilon\in(0,1)$; this is the \emph{sparsity} of $H$. 

Let $S\subseteq \mathbb{P}_n$ be the set of the indices $s$ of terms $P_s$ in $H$ with coefficients $|\mu_s|>\epsilon$. 
For $j\in\{0,1,2,\ldots,\left\lceil \log_2(1/\epsilon) \right\rceil\}$, define the set $S_j$ of all indices $s$ of terms $P_s$ in $H$ such that their coefficients $\mu_s$ satisfy $|\mu_s|\leq 2^{-j}$:
\begin{equation}
\label{eq:S_j}
    S_j = \{s: s\in S, |\mu_s|\leq 2^{-j}\},\quad j\in\{0,1,2,\ldots,\left\lceil \log_2(1/\epsilon) \right\rceil\},
\end{equation}
note that $S=S_0$. Let the set $\mathfrak{S}_j = S_j\setminus S_{j+1}$ be the set of all indices $s$ of terms $P_s$ in $H$ such that their coefficients $\mu_s$ satisfy $2^{-(j+1)}<|\mu_s|\leq 2^{-j}$:
\begin{equation}
    \mathfrak{S}_j = \{s: s\in S, 2^{-(j+1)} < |\mu_s|\leq 2^{-j}\},\quad j\in\{0,1,2,\ldots,\left\lceil \log_2(1/\epsilon) \right\rceil-1\}.
\end{equation}

Let $U_H(t) = e^{-iHt}$ to represent the time evolution operator under $H$.
Let $\ket{\EPR} = \frac{1}{\sqrt{2}}(\ket{00}+\ket{11})$ be the $2$-qubit EPR state, and let $\ket{\EPR_n} = \bigotimes_{j=1}^n\ket{\EPR}$ be the $2n$-qubit EPR state.

Consider a general quantum channel $\Lambda$.  
One can write it in its Kraus operator expression as
\begin{equation}
\label{eq:Kraus_operator}
    \Lambda (\rho) = \sum_kK_k\rho K_k^\dagger,
\end{equation}
where $\sum_k K_k^\dagger K_k = I$. 

Expanding the Kraus operator in the Pauli basis, we find
\begin{equation}
\label{eq:Pauli_expansion_of_Kraus}
    K_j = \sum_{\sigma_{k}\in\{I,\sigma_x,\sigma_y,\sigma_z\}} \alpha_{j,k}\sigma_{k}.
\end{equation}

\section{Main results and proof ideas\label{app:main_ancillary}}
The main result of Hamiltonian learning with ancillary systems can be summarized as
\begin{theorem}[2-copy Heisenberg-limited Hamiltonian learning algorithm]
\label{thm:2-copy_learning_alg}
For an arbitrary $n$-qubit unknown Hamiltonian $H=\sum_{s\in S}\mu_s P_{s}$, with $|\mu_s|\leq 1$ and $|S|\leq M$, there exists a hierarchical learning quantum algorithm which only queries the black box forward evolution of $H$, and a fault-tolerant quantum computer with $n$ ancillary qubits that outputs a classical description $\hat{\mu}_s$ of $\mu_s$ such that $|\hat{\mu}_s-\mu_s|<\epsilon$ with probability at least $1-\delta$, and the total experimental time is 

$$
T = \mathcal{O}\left(\frac{M^2\log(M/\delta)[\log(1/\epsilon)]^2}{\epsilon}\right).
$$
This algorithm requires no non-trivial classical post-processing and is robust against SPAM error.
\end{theorem}

In addition, we proposed a single-copy Hamiltonian learning algorithm without an ancillary system using only single-qubit operations, which can be summarized as:
\begin{theorem}[Single-copy Heisenberg-limited Hamiltonian learning algorithm]
\label{thm:1-copy_learning_alg}
For an arbitrary $n$-qubit unknown Hamiltonian $H=\sum_{s\in S}\mu_s P_{s}$, with $|\mu_s|\leq 1$ and $|S|\leq M$, there exists a hierarchical learning quantum algorithm which only queries the black box forward evolution of $H$, and a fault-tolerant quantum computer with no ancillary qubits that outputs a classical description $\hat{\mu}_s$ of $\mu_s$ such that $|\hat{\mu}_s-\mu_s|<\epsilon$ with probability at least $1-\delta$, and the total experimental time is 
$$
T = \mathcal{O}\left(\frac{M^3(\log(Mn)\log(1/\epsilon) + \log(M/\delta)[\log(1/\epsilon)]^2)}{\epsilon}\right).
$$
This algorithm requires a total classical post-processing time
$$
    T^C = \mathcal{O}\left(M^5n\log(Mn/\delta)\log(1/\epsilon)\right)
$$
\end{theorem}

The main subroutine of our learning algorithms is defined as follows:
\begin{defn}[Hierarchical learning subroutine]
    Given 1. an unknown Hamiltonian $H=\sum_{s\in S} \mu_sP_s$, with $|\mu_s|\leq 1$ and $|S|\leq M$, and 2.the classical description of $\hat{H}_j$ (its Pauli operator terms and coefficients), an estimation of all coefficients larger than $2^{-j}$ up to $\epsilon$ term-wise error. There exists a hierarchical learning subroutine $\mathcal{A}_j$ to obtain $\hat{H}_{j+1}$ which is an estimation of all coefficients larger than $2^{-(j+1)}$ up to $\epsilon$ term-wise error, for $j = 0,1,2,\ldots,\lceil\log_2(1/\epsilon)\rceil-1$. The final estimation $\hat{H}_{\lceil\log_2(1/\epsilon)\rceil-1}$ will be a characterization of $H$ to $\epsilon$ term-wise error.
\end{defn}

\textbf{The main ideas: }Based on the definitions provided in the previous section, we partition all unknown coefficients into the following sets: $S=\mathfrak{S}_{0}\cup \mathfrak{S}_{1}\cup\dots\cup \mathfrak{S}_{\lceil\log_2(1/\epsilon)\rceil-1}$, where $\mathfrak{S}_j = \{s: s\in S, 2^{-(j+1)} < |\mu_s|\leq 2^{-j}\}$. At each step $j$, the hierarchical learning subroutine $\mathcal{A}_j$ first identifies the index $s \in \mathfrak{S}_j$ with high probability using the Hamiltonian structure learning algorithm  $\mathcal{A}_j^{I}$. With this information, the corresponding coefficients can then be learned using the Hamiltonian coefficient learning algorithm  $\mathcal{A}_j^{II}$. Below, we provide an overview of the key components of both algorithms, with a detailed analysis deferred to the subsequent sections.

In the Hamiltonian structure learning $\mathcal{A}_{j}^{I}$, we first create the EPR state between unknown system with $N$-qubit ancillary system using transversal gates. Then we perform the time evolution of the unknown system under $\tilde{H}_j = (H-\hat{H}_j)/2^{-j}$ for time $t$ by querying unknown system and a fault-tolerant quantum computer iteratively. Lastly, we perform the Bell-basis measurement on the $2N$-qubit system. We prove that if $t=\mathcal{O}\left(\frac{1}{CM}\right)$ where $C$ is a big constant to supress the error, then all elements $s=\mathfrak{S}_j$ will be sampled once with high probability. There are several approximation errors involved in this step: 1. Trotterization error ($\varepsilon_{\mathrm{Trotter}}$) for simulating time-evolution under $\tilde{H}_j$, 2. truncation error ($\varepsilon_{\mathrm{Taylor}}$) from the Taylor expansion of time $t=\mathcal{O}\left(\frac{1}{CM}\right)$ evolution, and SPAM errors ($\varepsilon_{\mathrm{SPAM}}$). In \hyperref[lem:lower_bound_for_prob_dist]{Lemma 3}, we prove that for elements in $\mathfrak{S}_j$, their probability is lower bounded by 
$\widetilde{\Omega}\left(\frac{1}{4C^2M^2}-\frac{1}{C^3M^2}-\frac{2^{2j}}{C^2r}-\varepsilon_{\text{SPAM}}\right)$, where $r$ is the Trotter steps. Then one can sample all elements in $\mathfrak{S}_j$ at least once with probability $1-\delta$ by querying this algorithm $\mathcal{O}(C^2M^2\log(M/\delta))$ times. Then the total evolution time of all $\mathcal{A}^{I}_j$ is $\mathcal{O}\left(\frac{M \log(M/\delta)}{\epsilon}\right)$, which reaches the Heisenberg-limited scaling.

Then we use the Hamiltonian coefficient learning $\mathcal{A}^{II}_j$ to learn the coefficients $\mu_s$ for $s$ in the set $\widetilde{\mathfrak{S}_{j}}$, which we identified in the previous step. There are at most $|\widetilde{\mathfrak{S}_{j}}|\equiv L_j=\mathcal{O}(M^2\log(M/\delta))$ terms sampled in the previous step. We first use the Hamiltonian reshaping technique as introduced in \Cref{sec:app_Hamiltonian_reshaping} to approximate the time-evolution under a single term $P_s$ for $s\in \widetilde{\mathfrak{S}_{j}}$. We show the approximation error is bounded by $\mathcal{O}(\frac{M^2t^2}{r_2})$ for total evolution time $t$ with $r_2$ Trotter steps. Then we adapt the robust frequency estimation (RFE) as introduced in \Cref{sec:robust_frequency_estimation_details} to estimate $\mu_s$ within $\epsilon$ accuracy. In Theorem \ref{thm:robust_frequency_estimation}, we show the total evolution time in each RFE procedure is $\mathcal{O}(\frac{\log(1/\epsilon)+\log \log(2^{-j}/\epsilon)}{\epsilon})$. As there are at most $L_j$ terms, the total evolution time for $\mathcal{A}_j^{II}$ is $\mathcal{O}(\frac{M^2\log(M/\delta)(\log(1/\epsilon)+\log \log(2^{-j}/\epsilon))}{\epsilon})$, which also reaches the Heisenberg-limited scaling. Moreover, since we estimate each coefficient at a time, the effective Hamiltonian $H_{eff}$ after reshaping is a single term Hamiltonian, we only need product state input in RFE instead of highly entangled Bell state input as in \cite{ma2024learningkbodyhamiltonianscompressed}.

In Section \ref{sec:single-copy}, we investigate whether quantum entanglement of the 2-copy in structure learning is necessary. Surprisingly, we found a novel structure learning algorithm that only uses product state input without any ancillary qubits that also achieves the Heisenberg limit when used in the Hierarchical learning subroutine.

In Section \ref{app:lower}, we provide rigorous proofs for the lower bound of the general Hamiltonian learning algorithm in \Cref{thm:lower}, where we utilized the learning tree representation and applied Le Cam's two-point method for Hamiltonian distinguishing problem. The lower bound implies that for a Hamiltonian learning algorithm, there is a trade-off between the total evolution time complexity and the number of quantum controls.

\section{Hamiltonian structure learning \texorpdfstring{$\mathcal{A}^I_j$}{A I j}\label{app:structure_learning}}
We here provide details for Hamiltonian structure learning $\mathcal{A}^I_j$. In $\mathcal{A}^I_j$, we first use Trotterization to approximate the time evolution channel under $\tilde{H}_j = (H-\hat{H}_j)/2^{-j}$. Then, we evolve an $N$-qubit EPR state into this channel and perform a Bell-basis measurement at the end. By analyzing the Taylor expansion of the time evolution operator, and analyzing how considering the probability distribution of Bell-basis measurement will deviate if only considering the first-order truncation. Taking the Trotterization error $\varepsilon_{\text{Trotter}}$, first-order truncation error $\varepsilon_{\text{Taylor}}$, error caused by the imperfections of $\hat{H}_j$, and the state-preparation and measurement (SPAM) error $\varepsilon_{\text{SPAM}}$ into consideration, we show that $\mathcal{A}^I_j$ can identify all the terms in $H$ that is larger than $2^{-(j+1)}$ with high probability.

\subsection{Hamiltonian simulation with Trotterization\label{app:Trotter}}
\label{sec:Hamiltonian_simulation_with_Trotterization}
By the Lie product formula, the time evolution operator $e^{-i\tilde{H}_jt}$ can be approximated by interleaving the time evolution operator $e^{-iHt/(2^{-j}r)}$ and $e^{i\hat{H}_jt/(2^{-j}r)}$ in the asymptotic limit
\begin{equation}
     e^{-i\tilde{H}_jt} = e^{-i(H-\hat{H}_j)t/2^{-j}} = \lim_{r\rightarrow \infty} \left(e^{-iHt/(2^{-j}r)}e^{i\hat{H}_jt/(2^{-j}r)}\right)^r,
\end{equation}
where $e^{i\hat{H}_jt/(2^{-j}r)}$ is realizable through Hamiltonian simulation since the classical characterization of $\hat{H}_j$ is known.

We consider $r$ is a finite integer, then by the Trotter-Suzuki formula \cite{Suzuki1, SuzukiConvergence, PRXQuantum.5.020330, PhysRevX.11.011020}, we have
\begin{equation}
\label{eq:trotter_evolution}
    e^{-i\tilde{H}_jt} = e^{-i(H-\hat{H}_j)t/2^{-j}} =  \left(e^{-iHt/(2^{-j}r)}e^{i\hat{H}_jt/(2^{-j}r)}\right)^r + \mathcal{O}\left([H, \hat{H}_j]\frac{(t/2^{-j})^2}{r}\right).
\end{equation}
To characterize the accuracy of approximating the time evolution channel, denote the time evolution channel under the Hamiltonian $\tilde{H}_j$ as
\begin{equation}
    \mathcal{U}_{\tilde{H}_j,t}(\rho) = e^{-i\tilde{H}_jt}\rho e^{i\tilde{H}_jt},
\end{equation}
and the first-order Trotter-Suzuki time-evolution unitary a Hamiltonian $H - \hat{H}_j$ as
\begin{equation}
    U^{(1)}_{\text{TS},\tilde{H}_j}(t) = U^{(1)}_{\text{TS},H - \hat{H}_j}(t/2^{-j}) = e^{-iHt/2^{-j}}e^{i\hat{H}_jt/2^{-j}}
\end{equation}
with the first-order Trotter-Suzuki time-evolution channel be
\begin{equation}
    \mathcal{U}^{(1)}_{\text{TS},H - \hat{H}_j,t/2^{-j}}(\rho) = U^{(1)}_{\text{TS},H - \hat{H}_j}(t/2^{-j})\rho U^{(1),\dagger}_{\text{TS},H - \hat{H}_j}(t/2^{-j}).
\end{equation}
By \eqref{eq:trotter_evolution}, we can bound the first-order approximation error $\varepsilon^{(1)}_{\text{Trotter}}$,
\begin{equation}
    \varepsilon^{(1)}_{\text{Trotter}} = \|\mathcal{U}_{\tilde{H}_j,t}-(\mathcal{U}^{(1)}_{\text{TS},H - \hat{H}_j,t/(2^{-j}r)})^r\|_{\diamond} = \mathcal{O}\left([\tilde{H}_j, H - \hat{H}_j]\frac{(t/2^{-j})^2}{r}\right),
\end{equation}
where $\|\cdot\|_{\diamond}$ is the diamond norm. Note that if both $H - \hat{H}_j$ and $\tilde{H}_j$ has at most $M$ Pauli terms with coefficients bounded by $1$, then $\varepsilon_{\text{Trotter}}$ can be bounded by
\begin{equation}
\label{eq:eps_trotter}
    \varepsilon_{\text{Trotter}} = \mathcal{O}\left(\frac{M^2(t/2^{-j})^2}{r}\right).
\end{equation}

\subsection{Taylor expansion of the time evolution operator\label{app:Taylor}}

Consider the Taylor expansion of the time evolution operator under $\tilde{H}_j$:
\begin{equation}
\label{eq:Taylor_expansion}
    e^{-i\tilde{H}_jt} = e^{-i(H-\hat{H}_j)t/2^{-j}} = I - it\frac{(H-\hat{H}_j)}{2^{-j}} - \frac{t^2}{2}\left(\frac{(H-\hat{H}_j)}{2^{-j}}\right)^2 + \mathcal{O}(t^3).
\end{equation}
Recall that $H = \sum_s\mu_sP_s$ and we only care about the coefficients such that $|\mu_s|\geq\epsilon$, thus we can express $H-\hat{H}_j$ as a rescaled Hamiltonian
\begin{equation}
\label{eq:tilde_H_j}
    \tilde{H}_j = \frac{H-\hat{H}_j}{2^{-j}} = \sum_{s\in S_j}\frac{\mu_s}{2^{-j}}P_s + \sum_{s\in S_0\setminus S_j}\frac{\mu_s-\hat{\mu}_s}{2^{-j}}P_s,
\end{equation}
where $S_j$ is defined in \eqref{eq:S_j}, and $\hat{\mu}_s$ are the estimated coefficients in $\hat{H}_j$. Note that the total number of Pauli terms in $\tilde{H}_j$ would still be bounded by $M$ and by definition $|\mu_s-\hat{\mu}_s|\leq\epsilon$, indicating that all rescaled coefficients $\frac{\mu_s-\hat{\mu}_s}{2^{-j}}$ in the second terms on the right-hand side of \Cref{eq:tilde_H_j} with $s\in S_0\setminus S_j$ will be upper bounded by $1/2$ since $\hat{\mu}_s$ is at most $\epsilon$ away from the terms in $H_j$ with $s\in S_0\setminus S_j$ in the $\ell^\infty$-norm, and $2^j\geq 1/2\cdot\epsilon$. This means that the error caused by inaccurate estimation of the terms in $\hat{H}_j$ will still be bounded even after rescaling. With this, we can rewrite \eqref{eq:Taylor_expansion} as
\begin{equation}
\begin{split}
\label{eq:Taylor_expansion_terms}
    e^{-i\tilde{H}_jt} = &I - it\left(\sum_{s\in S_j}\frac{\mu_s}{2^{-j}}P_s + \sum_{s\in S_0\setminus S_j}\frac{\mu_s-\hat{\mu}_s}{2^{-j}}P_s\right) - \frac{t^2}{2}\Bigg(\sum_{s,s'\in S_j}\frac{\mu_s\mu_{s'}}{2^{-2j}}P_sP_{s'} + \sum_{\substack{s\in S_j\\s'\in S_0\setminus S_j}}\frac{\mu_s(\mu_{s'}-\hat{\mu}_{s'})}{2^{-2j}}P_sP_{s'}\\
    &+ \sum_{\substack{s'\in S_j\\s\in S_0\setminus S_j}}\frac{\mu_{s'}(\mu_s-\hat{\mu}_s)}{2^{-2j}}P_sP_{s'} + \sum_{s,s'\in S_0\setminus S_j}\frac{(\mu_s-\hat{\mu}_s)(\mu_{s'}-\hat{\mu}_{s'})}{2^{-2j}}P_sP_{s'}\Bigg) + o(t^3).
\end{split}
\end{equation}

\subsection{Bell sampling}
\label{sec:Bell_sampling}
Consider performing Bell sampling to the time evolution under $e^{-i\tilde{H}_jt}$, with a $n$-qubit ancillary system. Input the $2n$-qubit EPR state into the $n$-qubit time evolution channel tensor a $n$-qubit identity channel, by \eqref{eq:Taylor_expansion_terms}, the output state is:
\begin{equation}
\begin{split}
    \ket{\psi_{\text{out}}} =& \left(e^{-i\tilde{H}_jt}\otimes I_n\right)\ket{\EPR_n}\\
    =& \left(\left(I_n - it\tilde{H}_j - \frac{t^2}{2}\tilde{H}_j^{2} + o(t^3) \right)\otimes I_n\right)\ket{\EPR_n}\\
    =& \left(\left(I_n - it\left(\sum_{s\in S_j}\frac{\mu_s}{2^{-j}}P_s + \sum_{s\in S_0\setminus S_j}\frac{\mu_s-\hat{\mu}_s}{2^{-j}}P_s\right) - \frac{t^2}{2}\tilde{H}_j^{(2)} + o(t^3) \right)\otimes I_n\right)\ket{\EPR_n}\\
    =& \ket{\EPR_n} - it\left(\sum_{s\in S_j}\frac{\mu_s}{2^{-j}}\bigotimes_{j=1}^n \left((\sigma_{s,j}\otimes I)\ket{\EPR}\right) + \sum_{s\in S_0\setminus S_j}\frac{\mu_s-\hat{\mu}_s}{2^{-j}}\bigotimes_{j=1}^n \left((\sigma_{s,j}\otimes I)\ket{\EPR}\right)\right)\\
    &+ \left(\left(- \frac{t^2}{2}\tilde{H}_j^{2} + o(t^3)\right)\otimes I_n\right)\ket{\EPR_n}
\end{split}
\end{equation}
where
\begin{equation}
    \tilde{H}_j = \left(\sum_{s\in S_j}\frac{\mu_s}{2^{-j}}P_s + \sum_{s\in S_0\setminus S_j}\frac{\mu_s-\hat{\mu}_s}{2^{-j}}P_s\right)
\end{equation}
and
\begin{equation}
\begin{aligned}
     \tilde{H}_j^{2} &= \Bigg(\sum_{s,s'\in S_j}\frac{\mu_s\mu_{s'}}{2^{-2j}}P_sP_{s'} + \sum_{\substack{s\in S_j\\s'\in S_0\setminus S_j}}\frac{\mu_s(\mu_{s'}-\hat{\mu}_{s'})}{2^{-2j}}P_sP_{s'} \\
     &+ \sum_{\substack{s'\in S_j\\s\in S_0\setminus S_j}}\frac{\mu_{s'}(\mu_s-\hat{\mu}_s)}{2^{-2j}}P_sP_{s'} + \sum_{s,s'\in S_0\setminus S_j}\frac{(\mu_s-\hat{\mu}_s)(\mu_{s'}-\hat{\mu}_{s'})}{2^{-2j}}P_sP_{s'}\Bigg)
\end{aligned}
\end{equation}
are the second-order Pauli terms.
Consider measuring $\ket{\psi_{\text{out}}}$ in the Bell basis states, and denote the probability distributions of the measurement outcomes as $\mathcal{D}_{\text{Bell}}$. 

Notice that $\mathcal{D}_{\text{Bell}}$ is a discrete probability distribution supporting on the $4^n$ indices set of $\mathbb{P}_n$, each corresponding to a $n$-qubit Pauli operator. The goal is to sample all indices in $\mathfrak{S}_j$ from $\mathcal{D}_{\text{Bell}}$ efficiently. As we will show in later sections, the probability of elements in $\mathfrak{S}_j$ will be lower bounded by an inverse polynomial of $M$.

Intuitively, if one only considers the first order terms in $\ket{\psi_{\text{out}}}$, the probability of those elements corresponding to a term $P_s$ in $H$ is approximately $t^2\mu_s^2/2^{-2j}$, and the probability of those elements correspond to a traceless Pauli operator $P$ not in $H$ is zero. However, one needs to take the higher order terms into consideration, and the probability distribution $\mathcal{D}_{\text{Bell}}$ is not entirely supported on the $M$ elements correspond to the $M$ Pauli terms in $H$.

Fortunately, our protocol does not require estimating this probability distribution to high accuracy, it works well as long as the probability of those terms with indices in $\mathfrak{S}_j$ are at least inverse polynomials of $M$. Consider the second order terms, the absolute values of all coefficients in $\tilde{H}_j^{2}$ are upper bounded by $1$. For each term $P_s$ with $s\in\mathfrak{S}_j$, there are at most $M/2$ pairs of $\{s_1,s_2\}$ such that $P_{s_1}P_{s_2}=e^{i\theta}P_s$ in the second-order terms, where $\theta$ is a phase. Similarly, for the $l$-th order terms there will be at most $\mathcal{O}(M^{l-1})$ $l$-element sets $\{s_a,s_2,\ldots,s_l\}$ such that $\prod_{k=1}^l P_{s_k}=e^{i\theta}P_s$, where $\theta$ is a phase. Thus, for an index $s\in\mathfrak{S}_j$, the probability of sampling this element is lower bounded by
\begin{equation}
    \Omega\left(\left(\frac{t}{2} - \sum_{l=2}^\infty\left(M^{l-1}t^l\right)\right)^2\right),
\end{equation}
where $\frac{1}{2}$ is the lower bound of $\mu_s/2^{-j}$ for $s\in\mathfrak{S}_j$, and for higher-ordered terms we use the upper bound $1$ for all coefficients. 
Note that for the terms $s\in S_0\setminus S_j$, their coefficients are the remanent $|\mu_s-\hat{\mu}_s|$ caused by the inaccurate estimations in previous steps will still be upper bounded by $\epsilon/2^{-{j}}\leq 1/2$ and do not contribute more than the other terms.
If 
\begin{equation}
\label{eq:upper_bound_for_t}
    t < \frac{1}{CM},
\end{equation}
where $C$ is a large constant, the probability will be lower bounded by
\begin{equation}
\begin{split}
    &\Omega\left(\left(\frac{1}{2CM}-\sum_{l=2}^\infty\frac{1}{C^lM}\right)^2\right)\\
    =&\Omega\left(\left(\frac{1}{2CM}-\frac{1}{C^2M}\right)^2\right)\\
    =&\Omega\left(\frac{1}{4C^2M^2}-\frac{1}{C^3M^2} + \frac{1}{C^4M^2}\right)\\
    =&\Omega\left(\frac{1}{4C^2M^2}-\frac{1}{C^3M^2}\right).
\end{split}
\end{equation}
Thus, of the measurement outcome will deviate from $\frac{1}{4C^2M^2}$ by at most:
\begin{equation}
\label{eq:eps_Taylor}
    \varepsilon_{\text{Taylor}} = \mathcal{O}\left(\frac{1}{C^3M^2}\right).
\end{equation}

\subsection{SPAM error}
\label{sec:SPAM_error}
We define the state preparation and measurement (SPAM) error in the following way,
\begin{defn}
    \label{defn:SPAM_err}
    Take the noise in the preparation of the initial state as an error channel $\mathcal{E}_{\mathrm{prep}}$ applied after the ideal state preparation channel, and the noise of measurement as an error channel $\mathcal{E}_{\mathrm{meas}}$  applied before the ideal measurement channel.
    We assume that
    \[
    \|\mathcal{E}_{\mathrm{prep}}-\mathcal{I}\|_{\diamond}+\|\mathcal{E}_{\mathrm{meas}}-\mathcal{I}\|_{\diamond}\leq \varepsilon_{\mathrm{SPAM}},
    \]
    where $\varepsilon_{\mathrm{SPAM}}$ is the bound of the SPAM error.
\end{defn}
Therefore, consider a ideal quantum channel $\mathcal{E}$, the channel with SPAM error can be written as:
\begin{equation}
    \widetilde{\mathcal{E}}= \mathcal{E}_{\mathrm{meas}}\circ \mathcal{E}\circ\mathcal{E}_{\text{prep}},
\end{equation}
and 
\begin{equation}
\begin{split}
    &\|\mathcal{E} - \widetilde{\mathcal{E}}\|_{\diamond}\\
    = & \|\mathcal{I}\circ\mathcal{E}\circ\mathcal{I} - \mathcal{E}_{\text{meas}}\circ\mathcal{E}\circ\mathcal{E}_{\text{prep}}\|_{\diamond}\\
    = & \|\mathcal{I}\circ\mathcal{E}\circ\mathcal{I} - \mathcal{E}_{\text{meas}}\circ\mathcal{E}\circ\mathcal{I} + \mathcal{E}_{\text{meas}}\circ\mathcal{E}\circ\mathcal{I}- \mathcal{E}_{\text{meas}}\circ\mathcal{E}\circ\mathcal{E}_{\text{prep}}\|_{\diamond}\\
    = & \|\left(\mathcal{I}-\mathcal{E}_{\text{meas}}\right)\circ\mathcal{E}\circ\mathcal{I} + \mathcal{E}_{\text{meas}}\circ\mathcal{E}\circ\left(\mathcal{I}-\mathcal{E}_{\text{prep}}\right)\|_{\diamond}\\
    \leq & \|\left(\mathcal{I}-\mathcal{E}_{\text{meas}}\right)\circ\mathcal{E}\circ\mathcal{I}\|_{\diamond} + \|\mathcal{E}_{\text{meas}}\circ\mathcal{E}\circ\left(\mathcal{I}-\mathcal{E}_{\text{prep}}\right)\|_{\diamond}\\
    = & \|\left(\mathcal{I}-\mathcal{E}_{\text{meas}}\right)\circ\mathcal{E}\circ\mathcal{I}\|_{\diamond} + \|\left(\mathcal{E}_{\text{meas}}-\mathcal{I}\right)\circ\mathcal{E}\circ\left(\mathcal{I}-\mathcal{E}_{\text{prep}}\right) + \mathcal{I}\circ\mathcal{E}\circ\left(\mathcal{I}-\mathcal{E}_{\text{prep}}\right)\|_{\diamond}\\
    \leq & \|\left(\mathcal{I}-\mathcal{E}_{\text{meas}}\right)\circ\mathcal{E}\circ\mathcal{I}\|_{\diamond} + \|\mathcal{I}\circ\mathcal{E}\circ\left(\mathcal{I}-\mathcal{E}_{\text{prep}}\right)\|_{\diamond} + \|\left(\mathcal{E}_{\text{meas}}-\mathcal{I}\right)\circ\mathcal{E}\circ\left(\mathcal{I}-\mathcal{E}_{\text{prep}}\right)\|_{\diamond}\\
    \leq & \|\mathcal{I}-\mathcal{E}_{\text{meas}}\|_{\diamond}\cdot\|\mathcal{E}\|_{\diamond} + \|\mathcal{E}\|_{\diamond}\cdot\|\mathcal{I}-\mathcal{E}_{\text{prep}}\|_{\diamond} + \|\mathcal{E}_{\text{meas}}-\mathcal{I}\|_{\diamond}\cdot\|\mathcal{E}\|_{\diamond}\cdot\|\mathcal{I}-\mathcal{E}_{\text{prep}}\|_{\diamond}\\
    = & \left(\|\mathcal{I}-\mathcal{E}_{\text{meas}}\|_{\diamond} + \|\mathcal{I}-\mathcal{E}_{\text{prep}}\|_{\diamond} + \|\mathcal{I}-\mathcal{E}_{\text{meas}}\|_{\diamond}\cdot\|\mathcal{I}-\mathcal{E}_{\text{prep}}\|_{\diamond}\right)\cdot \|\mathcal{E}\|_{\diamond}\\
    \leq & \left(\varepsilon_{\text{SPAM}} + \frac{\varepsilon^2_{\text{SPAM}}}{4}\right)\|\mathcal{E}\|_{\diamond}\\
    \leq &\varepsilon_{\text{SPAM}} + \frac{\varepsilon^2_{\text{SPAM}}}{4}.
\end{split}
\end{equation}

\subsection{Lower bound for the probability of elements in \texorpdfstring{$\mathfrak{S}_j$}{S j} \label{sec:lower_bound_prob}}
In our protocol, we do not have direct access to the time evolution channel under $\tilde{H}_j$. Instead, we use the Trotterization method in \Cref{sec:Hamiltonian_simulation_with_Trotterization}. Moreover, SPAM errors in the experiments need to be considered. In analogy with the probability distribution $\widetilde{\mathcal{D}}_{\text{Bell}}$ as the probability distribution of Bell sampling in \Cref{sec:Bell_sampling}, define the probability distribution of doing Bell basis measurement on the output state of input $\ket{\EPR_n}$ to the Trotterization channel for simulation time evolution $e^{-i\tilde{H}_jt}$ as described in \Cref{sec:Hamiltonian_simulation_with_Trotterization} and $t<\frac{1}{CM}$ as in \eqref{eq:upper_bound_for_t}. Consider the probability of the elements in $\mathfrak{S}_j$, we provide the following lemma for a lower bound.

\begin{lemma}
\label{lem:lower_bound_for_prob_dist}
    For the probability distribution $\mathcal{D}_{\text{Bell}}$ defined on the $4^n$ indices set of $\mathbb{P}_n$, the probability of elements in $\mathfrak{S}_j$ is lower bounded by
    \begin{equation}
        \gamma_j = \widetilde{\Omega}\left(\frac{1}{4C^2M^2}-\frac{1}{C^3M^2}-\frac{2^{2j}}{C^2r}-\varepsilon_{\text{SPAM}}\right),
    \end{equation}
    where $C$ is a large constant as in \eqref{eq:upper_bound_for_t}, $M$ is the number of terms in $H$, $r$ is the number of steps in Trotterization,  $\varepsilon_{\text{SPAM}}$ is the SPAM error, and $\widetilde{\Omega}(f) = \Omega(f)-\mathcal{O}(\varepsilon_{\text{SPAM}}^2)$ only keep the leading order for the simplicity of the expression.
\end{lemma}
\begin{proof}
    We first set a target probability of any term with index $s\in\mathfrak{S}_j$ as
    \begin{equation}
        \tilde{\gamma}_j = \frac{\left(\min\{\mu_s:s\in\mathfrak{S}_j\}\right)^2\cdot(1/CM)^2}{2^{-2j}}.
    \end{equation}
By definition, $\mu_s/2^{-j}\geq 1/2$ for all $s\in\mathfrak{S}_j$, 
\begin{equation}
    \tilde{\gamma}_j = \Omega\left(\frac{1}{4C^2M^2}\right).
\end{equation}
Now we analyze how $\gamma_j$ deviate from $\tilde{\gamma}_j$. There are three sources that can contribute to the deviation, which are the Trotterization error as analyzed in \Cref{sec:Hamiltonian_simulation_with_Trotterization}, error originated from Bell sampling as in \Cref{sec:Bell_sampling}, and the SPAM error as in \Cref{sec:SPAM_error}. Notice that the Tortterization error $\varepsilon_{\text{Trotter}}$ in \eqref{eq:eps_trotter}, and the SPAM error $\varepsilon_{\text{SPAM}}$ are defined for the diamond distance between the ideal and actual quantum channels. By definition, the diamond distance directly implies an upper bound as the total variational distance between the probability distributions of the measurement outcome of the output state between the ideal and actual quantum channels with any input state, which is again by definition an upper bound of the difference between the probability on any element. Thus, the deviation caused by Trotterization and SPAM error is upper bounded by
\[
\mathcal{O}\left(\varepsilon_{\text{Trotter}}+\varepsilon_{\text{SPAM}}+\frac{\varepsilon_{\text{SPAM}^2}}{4}\right).
\]

Taking the error in Bell sampling $\varepsilon_{\text{Taylor}}$ as in \eqref{eq:eps_Taylor} into consideration, we can upper bound the difference
\begin{equation}
    |\tilde{\gamma}_j-\gamma_j| = \mathcal{O}\left(\frac{1}{C^3M^2}+\frac{2^{2j}}{C^2r}+\varepsilon_{\text{SPAM}}+\frac{\varepsilon_{\text{SPAM}}^2}{4}\right).
\end{equation}
Thus, we have
\begin{equation}
    \begin{split}
        \gamma_j &\geq \tilde{\gamma}_j - \mathcal{O}\left(\frac{1}{C^3M^2}+\frac{2^{2j}}{C^2r}+\varepsilon_{\text{SPAM}}+\frac{\varepsilon_{\text{SPAM}}^2}{4}\right)\\
        &= \widetilde{\Omega}\left(\frac{1}{4C^2M^2}-\frac{1}{C^3M^2}-\frac{2^{2j}}{C^2r}-\varepsilon_{\text{SPAM}}\right).
    \end{split}
\end{equation}
\end{proof}

Note that in the worst case $j=\left\lceil\log_2(1/\epsilon)\right\rceil$, and $2^{2j}\approx 1/\epsilon^2$, one need to choose $r = \mathcal{O}(1/\epsilon^2)$ to make the lower bound $\gamma_j$ independent of $\epsilon$.

\subsection{Determine the high-probability elements of a sparse probability distribution\label{sec:high_prob_detection}}

\begin{lemma}
\label{lem:sample_complexity_from_dist_with_gap}
    Consider a probability distribution $\mathcal{D}$ on a discrete space $\mathfrak{S}$ of size $N$, let the support of $\text{supp}(\mathcal{D})$ be the set of elements in $\mathfrak{S}$ on which $\mathcal{D}$ has non-zero probability, and let $M$ be the \emph{sparsity} of $\mathcal{D}$, i.e. $M = |\text{supp}(\mathcal{D})|\leq N$. Further define the set of high-probability elements $\mathfrak{S}_{\gamma}\subseteq \mathfrak{S}$ as:
\begin{equation}
    \text{supp}_\gamma(\mathcal{D}) = \{e\in\text{supp}(\mathcal{D}):\Pr(e)\geq\gamma\}.
\end{equation}
Note that $|\text{supp}_\gamma(\mathcal{D})|\leq |\text{supp}(\mathcal{D})| = M$. Then with $\mathrm{L} = \mathcal{O}\left(\frac{\log(M/\delta_{\gamma})}{\gamma}\right)$ samples, except for $1-\delta_{\gamma}$ probability, each element in $\text{supp}_\gamma(\mathcal{D})$ is sampled at least once.
\end{lemma}
\begin{proof}
If we independently sample from $\mathcal{D}$ for $L$ times, the probability that an element $j\in \text{supp}_\gamma(\mathcal{D})$ is not sampled for at least one time is upper bounded by $(1-\gamma)^{\mathrm{L}}$. If we want the probability that all $M$ elements in the support to be sampled for at least one time to exceed $1-\delta_{\gamma}$, using a union bound, we can set $\mathrm{L}$ to satisfy
\begin{equation}
    |\text{supp}_\gamma(\mathcal{D})|(1-\gamma)^{\mathrm{L}}\leq M(1-\gamma)^{\mathrm{L}}\leq \delta_{\gamma}.
\end{equation}

Thus, the required sample complexity to sample each element at least once for $1-\delta_{\gamma}$ probability is
\begin{equation}
\label{eq:sample_complexity_for_support_of_classical}
    \mathrm{L} = \mathcal{O}\left(\frac{\log(M/\delta_{\gamma})}{\gamma}\right).
\end{equation}
\end{proof}

\subsection{The complexities of \texorpdfstring{$\mathcal{A}^{I}_j$}{A I j}\label{sec:complexity_I}}
\label{sec:complexities_of_A_1}
Combining the lower bound $\gamma_j$ in \Cref{lem:lower_bound_for_prob_dist} and the sample complexity in \Cref{lem:sample_complexity_from_dist_with_gap}, we can directly give the sample complexity for sampling all indices $s\in\mathfrak{S}_j$ with probability at least $\delta_j$ as
\begin{equation}
    L_j = \mathcal{O}\left(\frac{\log(M/\delta_j)}{\gamma_j}\right)=\mathcal{O}\left(M^2\log(M/\delta_j)\right).
\end{equation}
For each sample, the total evolution time under $H$ is 
\begin{equation}
    t_{j,l} = \mathcal{O}\left(\frac{1}{CM}\cdot\frac{1}{2^{-j}}\right),
\end{equation}
and the total evolution time for $\mathcal{A}^{I}_j$ is
\begin{equation}
    T_{1,j} = \sum_{l=1}^{L_j} t_{j,l} = \mathcal{O}(2^j M\log(M/\delta_{j})),\quad j = \{0,1,2,\ldots, \left\lceil\log_2(1/\epsilon)\right\rceil\}.
\end{equation}
Note that in when $j=\left\lceil\log_2(1/\epsilon)\right\rceil$, 
\[
T_{1,\left\lceil\log_2(1/\epsilon)\right\rceil} = \mathcal{O}\left(\frac{M\log(M/\delta_{j})}{\epsilon}\right),
\]
which reaches the Heisenberg-limited scaling.

\section{Hamiltonian coefficient learning \texorpdfstring{$\mathcal{A}^{II}_j$}{A II j}}
\label{app:coef_learning}
In this section we provide proofs for learning Hamiltonian coefficients in $\mathcal{A}^{II}_j$, which constitutes of two parts Hamiltonian reshaping and robust frequency estimation as we will further illustrate in the followings. The basic idea is that once we identified all possible terms in $\mathfrak{S}_j$ from $\mathcal{A}^{I}_j$, there are at most $L_j$ possible terms sampled. Let the set of all sampled indices from $\mathcal{A}^{I}_j$ as $\mathfrak{S}_{L_j}$, with $|\mathfrak{S}_{L_j}|\leq L_j = \mathcal{O}\left(M^2\log(M/\delta_j)\right)$. We first use the Hamiltonian reshaping technique as introduced in \Cref{sec:app_Hamiltonian_reshaping} to approximate the time-evolution under a single term $P_s$ for $s\in \mathfrak{S}_{L_j}$. Later, with the time-evolution under the single term $P_s$, we adapt the robust frequency estimation protocol as introduced in \Cref{sec:robust_frequency_estimation_details} to estimate $\mu_s$ to $\epsilon$ accuracy.

\subsection{Hamiltonian reshaping}
\label{sec:app_Hamiltonian_reshaping}
In this work we will need to use the Hamiltonian reshaping technique to obtain an effective Hamiltonian that consists of only a single Pauli operator.
More precisely, given a Pauli operator $P_a\in\{I,X,Y,Z\}^{\otimes n}\setminus\{I^{\otimes n}\}$, we want the effective Hamiltonian to contain only this term, with its coefficient value the same as in the original Hamiltonian $H$. To achieve this, we apply, with an interval of $\tau$, Pauli operators randomly drawn from the set

\begin{equation}
    \label{eq:random_unitaries_ensemble}
    \mathcal{K}_{P_a} = \left\{P\in \mathbb{P}_n:[P_a,P_b]=0\right\},
\end{equation}
where $\mathbb{P}_n$ is the set of all $n$-qubit Pauli operators. Here we can see that $|\mathcal{K}_{P_a}| = 2^{2n-1}$ and $\mathcal{K}_{P_a}$ are the set of all $n$-qubit Pauli operators that commute with $P_a$.
Also, $\tau$ needs to be sufficiently small as will be analyzed in \Cref{thm:hamiltonian_reshaping}.
More precisely, the evolution of the quantum system is described by
\begin{equation}
    \label{eq:hamiltonian_reshaping}
    Q_{r_2} e^{-iH\tau}Q_{r_2}\cdots Q_2 e^{-iH\tau}Q_2Q_1 e^{-iH\tau}Q_1,
\end{equation}
where $Q_k$, $k=1,2,\cdots, r_2$, is uniformly randomly drawn from the set $\mathcal{K}_{P_a}$. In one time step of length $\tau$, the quantum state evolves under the quantum channel as
\begin{equation}
\begin{aligned}
    \rho &\mapsto \frac{1}{2^{2n-1}}\sum_{Q\in \mathcal{K}_{P_a}} Qe^{-iH\tau}Q\rho Qe^{iH\tau}Q 
    = \rho-\frac{1}{2^{2n-1}}\sum_{Q\in \mathcal{K}_{P_a}} i\tau [QHQ,\rho]+\mathcal{O}(\tau^2) \\
    &= \rho - i\tau [H_{\mathrm{eff}},\rho]+\mathcal{O}(\tau^2) 
    =e^{-iH_{\mathrm{eff}}\tau}\rho e^{iH_{\mathrm{eff}}\tau} + \mathcal{O}(\tau^2),
\end{aligned}
\end{equation}
where $H_{\mathrm{eff}}$, the effective Hamiltonian, is
\begin{equation}
    H_{\mathrm{eff}} = \frac{1}{2^{2n-1}}\sum_{Q\in \mathcal{K}_{P_a}} QHQ.
\end{equation}
The above can be interpreted as a linear transformation applied to the Hamiltonian $H$. Consequently, it is natural to examine the impact of this transformation on each Pauli term within the Hamiltonian. For a term $P$, we note that there are two possible outcomes:
\begin{lemma}
    \label{lem:Hamiltonian_reshaping_Pauli_term}
    Let $P$ be a Pauli operator and let $\mathcal{K}_{P_a}$ be as defined in \eqref{eq:random_unitaries_ensemble}. Then
    \begin{equation}
    \frac{1}{2^{2n-1}}\sum_{Q\in \mathcal{K}_{P_a}} QPQ=
    \begin{cases}
        P_a,\text{ if } P=P_a,\\
        0,\text{ if } P\neq P_a.
    \end{cases}
\end{equation}
\end{lemma}
\begin{proof}
    If $P=P_a$, then $QPQ=P_a$ for all $Q\in \mathcal{K}_{P_a}$, averaging over all $Q$ therefore yields $P_a$. If $P\neq P_a$, then there exists $Q_0\in \mathcal{K}_{P_a}$ such that $PQ_0=-Q_0P$ (because $\mathcal{K}_{P_a}$ contains all Pauli operators that commute with $P_a$ and two Pauli matrices either commute or anti-commute). Consider the mapping $\phi:\mathcal{K}_{P_a}\rightarrow\mathcal{K}_{P_a}$ defined by $\phi(Q) = Q_0Q$. This is a bijection from $\mathcal{K}_\beta$ to itself, and it can be readily checked that if $Q$ commutes with $P$, then $\phi(Q)$ anti-commutes with $P$; if $Q$ anti-commutes with $P$, then $\phi(Q)$ commutes with $P$. 
    Consequently $QPQ=P$ for half of all $Q\in \mathcal{K}_{P_a}$ and $QPQ=-P$ for the other half. 
    Thus taking the average yields $0$. 
\end{proof}
This lemma allows us to on average single out the Pauli term $P_a$ we want to preserve in the Hamiltonian and discard all other terms. We have:
\begin{equation}
    \label{eq:effective_hamiltonian_terms}
    H_{\mathrm{eff}} = \mu_a P_a.
\end{equation}
Note that the coefficients $\mu_a$ of $P_a$ are preserved in this effective Hamiltonian.

While Lemma~\ref{lem:Hamiltonian_reshaping_Pauli_term} concerns the uniform average over all $2^{2n-1}$ elements of a set of Pauli operators, in our learning protocol we will randomly sample from this set. To assess the protocol's accuracy, we will use the following theorem.

\begin{theorem}[Theorem~3 of \cite{ma2024learningkbodyhamiltonianscompressed}]
    \label{thm:hamiltonian_reshaping}
    Let $P_a$ traceless $n$-qubit Pauli operator, and let $\mathcal{K}_{P_a}$ be as defined in \eqref{eq:random_unitaries_ensemble}.
    Let $U$ be the random unitary defined in \eqref{eq:hamiltonian_reshaping}. Let $V=e^{-iH_{\mathrm{eff}}t}$ for $H_{\mathrm{eff}}$ given in \eqref{eq:effective_hamiltonian_terms}, and $t=r_2\tau$.
    We define the quantum channels $\mathcal{U}$ and $\mathcal{V}$ be
    \[
    \mathcal{U}(\rho) = \mathbb{E}[U\rho U^\dag],\quad \mathcal{V}(\rho) = V\rho V^\dag.
    \]
    Then 
    \[
    \|\mathcal{U}-\mathcal{V}\|_{\diamond}\leq \frac{4M^2 t^2}{r_2}.
    \]
\end{theorem}
Here the norm is the diamond norm (or completely-bounded norm), given in terms of the Schatten $L_1$ norm by $\|\mathcal{U}\|_\diamond := \max_{X : \|X\|_1 \le 1}\bigl\|\bigl(\mathcal{U} \otimes I_n\bigl)(X)\bigr\|_1$. 
\begin{proof}(Appendix~A of \cite{ma2024learningkbodyhamiltonianscompressed})

    Recall from \eqref{eq:hamiltonian_reshaping} that
    \[
    U = Q_r e^{-iH\tau} Q_r\cdots Q_1 e^{-iH\tau} Q_1.
    \]
    We define
    \[
    U_j = Q_j e^{-iH\tau} Q_j,
    \]
    and the quantum channel $\mathcal{U}_j$
    \[
    \mathcal{U}_j(\rho) = \mathbb{E}_{Q_j} U_j\rho U_j^\dag.
    \]
    Because each $Q_j$ is chosen independently, we have
    \[
    \mathcal{U} = \mathcal{U}_{r_2} \mathcal{U}_{r_2-1}\cdots \mathcal{U}_1.
    \]
    We then define $\mathcal{V}_\tau(\rho) = e^{-iH_{\mathrm{eff}}\tau}\rho e^{iH_{\mathrm{eff}}\tau}$, which then gives us $\mathcal{V} = \mathcal{V}_\tau^{r_2}$. We will then focus on obtaining a bound for $\|\mathcal{U}_j-\mathcal{V}_\tau\|_{\diamond}$.

    Consider a quantum state $\rho$ on the joint system consists of the current system and an auxiliary system denoted by $\alpha$, and an observable $O$ on the combined quantum system such that $\|O\|\leq 1$. We have 
    \[
    \begin{aligned}
        &\Tr[O (I_\alpha\otimes U_j)\rho (I_\alpha\otimes U_j^\dag)] \\
        &= \Tr[O (I_\alpha\otimes (Q_j e^{-iH\tau}Q_j))\rho (I_\alpha\otimes (Q_j e^{iH\tau}Q_j))] \\
        &= \Tr[O\rho] - i\tau \Tr[O[I_\alpha\otimes (Q_j H Q_j),\rho]] \\
        &+\tau^2\Tr[O(I_\alpha\otimes (Q_j H Q_j))\rho (I_\alpha\otimes (Q_j H Q_j))]\\
        &-\frac{\tau^2}{2}\Tr[O(I_\alpha\otimes (Q_j H^2 e^{-iH \xi}Q_j))\rho]-\frac{\tau^2}{2}\Tr[O\rho (I_\alpha\otimes (Q_j H^2 e^{+iH \xi}Q_j))],
    \end{aligned}
    \]
    where we are using Taylor's theorem with the remainder in the Lagrange form and $\xi\in [0,\tau]$. This tells us that
    \[
    \left|\Tr[O (I_\alpha\otimes U_j)\rho (I_\alpha\otimes U_j^\dag)] - (\Tr[O\rho] - i\tau \Tr[O[I_\alpha\otimes (Q_j H Q_j),\rho]] )\right|\leq 2\|H\|^2\tau^2
    \]
    Because the two terms on the left-hand side, when we take the expectation value with respect to $Q_j$, become
    \[
    \begin{aligned}
        &\mathbb{E}_{Q_j} \Tr[O (I_\alpha\otimes U_j)\rho (I_\alpha\otimes U_j^\dag)] = \Tr[O (\mathcal{I}_\alpha\otimes \mathcal{U}_j)(\rho)] \\
        &\mathbb{E}_{Q_j} \big[\Tr[O\rho] - i\tau \Tr[O[I_\alpha\otimes (Q_j H Q_j),\rho]] \big] = \Tr[O\rho] - i\tau \Tr[O[I_\alpha\otimes H_{\mathrm{eff}},\rho]],
    \end{aligned}
    \]
    we have
    \begin{equation*}
         \left|\Tr[O (\mathcal{I}_\alpha\otimes \mathcal{U}_j)(\rho)] - (\Tr[O\rho] - i\tau \Tr[O[I_\alpha\otimes H_{\mathrm{eff}},\rho]])\right|\leq 2\|H\|^2\tau^2.
    \end{equation*}
    Because this is true for all $O$ such that $\|O\|\leq 1$,
    \begin{equation}
    \label{eq:random_unitary_short_time_approx}
        \|(\mathcal{I}_\alpha\otimes \mathcal{U}_j)(\rho)-(\rho-i\tau [I_\alpha\otimes H_{\mathrm{eff}},\rho])\|_1\leq 2\|H\|^2\tau^2.
    \end{equation}
    Similarly we can also show that
    \begin{equation}
    \label{eq:effective_ham_short_time_approx}
        \|(\mathcal{I}_\alpha\otimes \mathcal{V}_\tau)(\rho)-(\rho-i\tau [I_\alpha\otimes H_{\mathrm{eff}},\rho])\|_1\leq 2\|H_{\mathrm{eff}}\|^2\tau^2.
    \end{equation}
    Combining \eqref{eq:random_unitary_short_time_approx} and \eqref{eq:effective_ham_short_time_approx}, and by the triangle inequality, we have
    \begin{equation*}
        \|(\mathcal{I}_\alpha\otimes \mathcal{U}_j)(\rho)-(\mathcal{I}_\alpha\otimes \mathcal{V}_\tau)(\rho)\|_1\leq 2\|H\|^2\tau^2 + 2\|H_{\mathrm{eff}}\|^2\tau^2.
    \end{equation*}
    Because the above is true for all auxiliary system $\alpha$ and for all quantum states on the combined quantum system, we have
    \begin{equation}
        \|\mathcal{U}_j-\mathcal{V}_\tau\|_{\diamond}\leq 2\|H\|^2\tau^2 + 2\|H_{\mathrm{eff}}\|^2\tau^2.
    \end{equation}

    Next, we observe that
    \[
    \mathcal{U}-\mathcal{V} = \sum_{k=1}^{r_2} \mathcal{U}_r\cdots \mathcal{U}_{k+1}(\mathcal{U}_k-\mathcal{V}_\tau)\mathcal{V}_\tau^{k-1}.
    \]
    Therefore
    \[
    \|\mathcal{U}-\mathcal{V}\|_\diamond\leq \sum_{k=1}^{r_2} \|\mathcal{U}_k-\mathcal{V}_\tau\|_{\diamond}\leq r_2(2\|H\|^2\tau^2 + 2\|H_{\mathrm{eff}}\|^2\tau^2).
    \]
    The bound in the Theorem follows by observing that $\|H\|,\|H_{\mathrm{eff}}\|\leq M$, and $\tau = t/r_2$.
\end{proof}

\subsection{Robust frequency estimation}
\label{sec:robust_frequency_estimation_details}
In this section we introduce the robust frequency estimation protocol in \cite{LiTongNiGefenYing2023heisenberg}, which largely follows the idea of the robust phase estimation protocol in \cite{KimmelLowYoder2015robust} and is also used in \cite{ma2024learningkbodyhamiltonianscompressed}. However, unlike in the previous works which require highly entangled Bell state input, as described in \Cref{sec:Coefficient-learning algorithm}, we only require product state input in this work, since the effective Hamiltonian as in \Cref{eq:effective_hamiltonian_terms} has only one Pauli term, which means that we only need to set the input state to have maximum sensitivity on one qubit. Here we restate the proofs in \cite{ma2024learningkbodyhamiltonianscompressed}.

\begin{theorem}[Robust frequency estimation \cite{ma2024learningkbodyhamiltonianscompressed}]
    \label{thm:robust_frequency_estimation}
    Let $\theta\in[-A,A]$.
    Let $X(t)$ and $Y(t)$ be independent random variables satisfying
    \begin{equation}
        \begin{aligned}
            &|X(t)-\cos(\theta t)|< 1/\sqrt{2}, \text{ with probability at least }2/3, \\
            &|Y(t)-\sin(\theta t)|< 1/\sqrt{2}, \text{ with probability at least }2/3.
        \end{aligned}
    \end{equation}
    Then with $K$ independent non-adaptive\footnote{By ``non-adaptive'' we mean that the choice of each $t_j$ does not depend on the value of $X(t_{j'})$ or $Y(t_j')$ for any $j'$.} samples $X(t_1),X(t_2),\cdots,X(t_K)$ and $Y(t_1),Y(t_2),\cdots,Y(t_K)$, $t_j\geq 0$, for
    \begin{equation}
        K=\mathcal{O}(\log(A/\epsilon)(\log(1/q)+\log\log(A/\epsilon))),
    \end{equation}
    \begin{equation}
    \begin{aligned}
        T=\sum_{j=1}^Kt_j=\mathcal{O}((1/\epsilon)(\log(1/q)+\log\log(A/\epsilon))),\quad \max_j t_j=\mathcal{O}(1/\epsilon),
    \end{aligned}
    \end{equation}
    we can obtain a random variable $\hat{\theta}$ such that
    \begin{equation}
        \Pr[|\hat{\theta}-\theta|>\epsilon]\leq q.
    \end{equation}
\end{theorem}

We modify the protocol in \cite{LiTongNiGefenYing2023heisenberg} because our goal is to obtain an accurate estimate with large probability rather than having an optimal mean-squared error scaling.
The key tool is the following lemma that allows us to incrementally refine the frequency estimate:

\begin{lemma}
    \label{lem:frequency_est_refine}
    Let $\theta\in[a,b]$.
    Let $Z(t)$ be a random variable such that 
    \begin{equation}
        \label{eq:rv_correctness_condition}
        |Z(t)-e^{i\theta t}|< 1/2.
    \end{equation}
    Then we can correctly distinguish between two overlapping cases $\theta\in[a,(a+2b)/3]$ and $\theta\in[(2a+b)/3,b]$ with one sample of $Z(\pi/(b-a))$. 
\end{lemma}

\begin{proof}
    These two situations can be distinguished by looking at the value of 
    \[
    f(\theta)=\sin\left(\frac{\pi}{b-a}\left(\theta-\frac{a+b}{2}\right)\right).
    \]
    We know from \eqref{eq:rv_correctness_condition} that
    \[
    \left|\Im\left(e^{-i\frac{(a+b)\pi}{2(b-a)}}Z(\pi/(b-a))\right)-f(\theta)\right|<1/2.
    \]
    where $t$ in \eqref{eq:rv_correctness_condition} is substituted by $\pi/(b-a)$ and a phase factor is added.
    
    If $$\Im\left(e^{-i\frac{(a+b)\pi}{2(b-a)}}Z(\pi/(b-a))\right)\leq 0,$$ then $f(\theta)<1/2$, which implies $\theta\in[a,(a+2b)/3]$. If $$\Im\left(e^{-i\frac{(a+b)\pi}{2(b-a)}}Z(\pi/(b-a))\right)> 0,$$ then $f(\theta)>-1/2$, which implies $\theta\in[(2a+b)/3,b]$.
\end{proof}

Using this lemma we will prove \Cref{thm:robust_frequency_estimation}, which we restate below:

\begin{theorem*}
    Let $\theta\in[-A,A]$.
    Let $X(t)$ and $Y(t)$ be independent random variables satisfying
    \begin{equation}
        \label{eq:rv_correctness_condition_prob}
        \begin{aligned}
            &|X(t)-\cos(\theta t)|< 1/\sqrt{2}, \text{ with probability at least }2/3, \\
            &|Y(t)-\sin(\theta t)|< 1/\sqrt{2}, \text{ with probability at least }2/3.
        \end{aligned}
    \end{equation}
    Then with independent non-adaptive samples $X(t_1),X(t_2),\cdots,X(t_K)$ and $Y(t_1),Y(t_2),\cdots,Y(t_K)$, $t_j\geq 0$, for
    \begin{equation}
        \label{eq:number_of_samples_RPE}
        K=\mathcal{O}(\log(A/\epsilon)(\log(1/q)+\log\log(A/\epsilon))),
    \end{equation}
    \begin{equation}
    \label{eq:total_evolution_time_RPE}
    \begin{aligned}
        T=\sum_{j=1}^Kt_j=\mathcal{O}((1/\epsilon)(\log(1/q)+\log\log(A/\epsilon))),\quad \max_j t_j=\mathcal{O}(1/\epsilon),
    \end{aligned}
    \end{equation}
    we can obtain a random variable $\hat{\theta}$ such that
    \begin{equation}
        \Pr[|\hat{\theta}-\theta|>\epsilon]\leq q.
    \end{equation}
\end{theorem*}

\begin{proof}
    We let $\lambda=\epsilon/3$.
    We build a random variable $S(t)$ satisfying \eqref{eq:rv_correctness_condition}, with which we will iteratively narrow down the interval $[a,b]$ containing $\theta$ until $|b-a|\leq 2\lambda$, at which point we choose $\hat{\theta}=(a+b)/2$. If $\theta\in[a,b]$ we will then ensure $|\hat{\theta}-\theta|\leq \lambda$. However, each iteration will involve some failure probability, which we will analyze later.

    To build the random variable $S(t)$, we first use $m$ independent samples of $X(t)$ and then take median $X_{\mathrm{median}}(t)$, which satisfies 
    \[
    |X_{\mathrm{median}}(t)-\cos(\theta t)|\leq 1/\sqrt{2}
    \]
    with probability at least $1-\delta/2$, where by \eqref{eq:rv_correctness_condition_prob} and the Chernoff bound 
    \[
    \delta = c_1 e^{-c_2 m},
    \]
    for some universal constant $c_1,c_2$.
    Similarly we can obtain $Y_{\mathrm{median}}(t)$ such that
    \[
    |Y_{\mathrm{median}}(t)-\cos(\theta t)|\leq 1/\sqrt{2}
    \]
    with probability at least $1-\delta/2$. With these medians we then define
    \[
    S(t) = X_{\mathrm{median}}(t) + iY_{\mathrm{median}}(t).
    \]
    This random variable satisfies
    \[
    |S(t)-e^{i\theta t}|\leq 1/2
    \]
    with probability at least $1-\delta$ using the union bound. It therefore allows us to solve the discrimination task in \Cref{lem:frequency_est_refine} with probability at least $1-\delta$.

    Whether each iteration proceeds correctly or not, the algorithm terminates after 
    \[
    L=\lceil \log_{3/2}(A/\lambda) \rceil
    \]
    iterations. In the $l$th iteration we sample $X(s_l)$ and $Y(s_l)$ where $s_l=(\pi/2A)(3/2)^{l-1}$. 
    We use $m$ samples of $X(s_l)$ and $Y(s_l)$ for computing the median in each iteration, and therefore the failure probability is at most $\delta=c_1 e^{-c_2 m}$. 
    With probability at most $L\delta$ using the union bound, one of the iteration fails. 
    In order to ensure the protocol succeeds with probability at least $1-q$, it suffices to let $\delta \leq q/L$, and it therefore suffices to choose
    \[
    m = \lceil c_3 \log(L/q) \rceil=\mathcal{O}(\log(1/q)+\log\log(A/\epsilon)).
    \]
    The total number of samples (for either $X(t)$ or $Y(t)$) is therefore as described in \eqref{eq:number_of_samples_RPE}.
    All the $t$ in each sample added together is
    \begin{equation}
        \label{eq:total_time}
        T = m\sum_{l=1}^{L}  s_l = \frac{m\pi}{2A}\sum_{l=1}^L \left(\frac{3}{2}\right)^{l-1} \leq \frac{3\pi m}{2\epsilon}=\mathcal{O}((1/\epsilon)(\log(1/q)+\log\log(A/\epsilon))),
    \end{equation} 
    thus giving us \eqref{eq:total_evolution_time_RPE}.
\end{proof}

\subsection{The experimental setup of \texorpdfstring{$\mathcal{A}^{II}_j$}{A II j}}
\label{sec:experimental_setup_of_A_2}
With the aforementioned lemmas, we now describe the experimental setup for robustly estimating the coefficient of each term with $\epsilon$ accuracy. In an experiment for learning the coefficient $\mu_s$ of $P_s$, let $\beta_s\in\{I,x,y,z\}^{\otimes n}$ satisfies $P_s = \bigotimes_{j=1}^n \sigma^{\beta_{s,j}}$, where $\beta_{s,j}$ is the $j$-th entries of $\beta_s$. Denote the set of qubit indices where $P_s$ non-trivially act on as $S_{P_s}$. As defined in \Cref{eq:pauli_lambdastates}, we set
\begin{equation}
    \begin{split}
        &\ket{1,I} = \ket{0},  \quad \ket{1,x} = \ket{+}, \quad \ket{1,y} = \frac{1}{\sqrt{2}}(\ket{0}+i\ket{1}), \quad \ket{1,z} = \ket{0},\\
        &\ket{-1,I} = \ket{0}, \quad \ket{-1,x} = \ket{-},\quad \ket{-1,y} = \frac{1}{\sqrt{2}}(\ket{0}-i\ket{1}), \quad \ket{-1,z} = \ket{1},
    \end{split}
\end{equation}
and we define the following $\pm 1$ eigenstates of $P_s$ as:
\begin{equation}
    \ket{1,\beta_s} = \bigotimes_{j=1}^n \ket{1,\beta_{s,j}},\quad \ket{-1,\beta_s} = \bigotimes_{j=1}^n \ket{-1,\beta_{s,j}},
\end{equation}

We set the input state as
\begin{equation}
\label{eq:initial_state}
        \ket{\phi^s_0} = \frac{1}{\sqrt{2}}(\ket{1,\beta_s} + \ket{-1,\beta_s}) = (\bigotimes_{j=1}^{s^\star-1} \ket{1,\beta_{s,j}}) \otimes \frac{1}{\sqrt{2}}\left(\ket{1,\beta_{s,s^\star}} + \ket{-1,\beta_{s,s^\star}}\right) \otimes (\bigotimes_{j=s^\star+1}^{n} \ket{1,\beta_{s,j}}),
\end{equation}
which is a product-state, and is the equal-weight superposition of two eigenstates of $H_{\mathrm{eff}} = \mu_s P_s$ and thus has maximal sensitivity to the time-evolution under $H_{\mathrm{eff}}$. Evolve $\ket{\phi^s_0}$ under the effective Hamiltonian $H_{\mathrm{eff}}$ as described in \Cref{sec:app_Hamiltonian_reshaping}, at time $t$ we will approximately obtain the state
\begin{equation}
\label{eq:exact_state_phase_estimation_experiment}
        \ket{\phi^s_t} = \frac{1}{\sqrt{2}}(e^{-i\mu_st}\ket{1,\beta_s} + e^{i\mu_st}\ket{-1,\beta_s}).
\end{equation}
In the end, we then measure $\ket{\phi^s_t}$ with observables
\begin{equation}
\label{eq:obs_RFE}
\begin{split}
    O^+_s &= (\bigotimes_{j=1}^{s^\star-1} \ket{1,\beta_{s,j}}) \otimes Q^+_{s,j^\star} \otimes (\bigotimes_{j=s^\star+1}^{n} \ket{1,\beta_{s,j}}),\\
    O^-_s &= (\bigotimes_{j=1}^{s^\star-1} \ket{1,\beta_{s,j}}) \otimes Q^-_{s,j^\star} \otimes (\bigotimes_{j=s^\star+1}^{n} \ket{1,\beta_{s,j}}),
\end{split}
\end{equation}
where $Q^+_{s,j^\star}$ and $Q^+_{s,j^\star}$ are chosen to be single-qubit Pauli operators such that
\begin{equation}
\label{eq:Q_j-star}
\begin{split}
    Q^+_{s,j^\star}\ket{1,\beta_{s,j^\star}} &= \ket{-1,\beta_{s,j^\star}},\quad Q^+_{s,j^\star}\ket{-1,\beta_{s,j^\star}} = \ket{1,\beta_{s,j^\star}};\\
    Q^-_{s,j^\star}\ket{1,\beta_{s,j^\star}} &= i\ket{-1,\beta_{s,j^\star}},\quad Q^-_{s,j^\star}\ket{-1,\beta_{s,j^\star}} = -i\ket{1,\beta_{s,j^\star}}.
\end{split}
\end{equation}
Therefore, the expectation value is
\begin{equation}
    \bra{\phi^s_t}O^+_s\ket{\phi^s_t} = \cos(2\mu_s t),\quad \bra{\phi^s_t}O^-_s\ket{\phi^s_t} = \sin(2\mu_s t).
\end{equation}
Note that in order to measure $O^+_s$ and $O^-_s$, we only need to measure the $j^\star$-th qubit since all other qubits are in the $+1$ eigenstates and would not affect the phase.

We summarize the above experiment in the following definition
\begin{defn}[Frequency estimation experiment $\mathcal{A}^{II}$]
\label{def:phase_estimation_experiment}
    We call the procedure below a $(s,t,\tau)$-phase estimation experiment:
    \begin{enumerate}
        \item Prepare the initial product-state $\ket{\phi^s_0}=\frac{1}{\sqrt{2}}(\ket{1,\beta_s} + \ket{-1,\beta_s})$ as in \eqref{eq:initial_state}.
        \item Let the system evolve for time $t$ while applying random Pauli operators from $\mathcal{K}_{P_s}$ (defined in \eqref{eq:random_unitaries_ensemble}) with interval $\tau$.
        \item Measure the observables $O^+_s$ or $O^-_s$ (by measuring the $j^\star$-th qubit in the $Q^+_{s,j^\star}$ or $Q^-_{s,j^\star}$ basis) as defined above to obtain a $\pm 1$ outcome.
    \end{enumerate}
\end{defn}
The goal of the above experiment is to estimate $\mu_s$, for we need the robust frequency estimation protocol as introduced in \Cref{sec:robust_frequency_estimation_details}. We will discuss how this is done and analyze the effect of the Hamiltonian reshaping error and the state preparation and measurement (SPAM) error as defined in \Cref{defn:SPAM_err} in the following.

Through a $(s,t,\tau)$-phase estimation experiment, if we measure $O^+_s$ in the end, we will obtain a random variable $\nu_s^+(t)\in{\pm 1}$. If we have the exact  state $\ket{\phi^s_t}$ as defined in \eqref{eq:exact_state_phase_estimation_experiment}, then we can have $\mathbb{E}[\nu_s^+(t)] = \cos(2\mu_s t)$. However, due to the Hamiltonian reshaping error and the SPAM error, we have

\begin{equation}
\label{eq:phase_estimation_experiment_expectation_deviation}
    |\mathbb{E}[\nu_s^+(t)]-\cos(2\mu_s t)|\leq \frac{4M^2 t^2}{r_2} + \epsilon_{\mathrm{SPAM}}.
\end{equation}
Notice that the first term on the right-hand side comes from the Hamiltonian reshaping error, where we set $t=r_2\tau$ as in Theorem~\ref{thm:hamiltonian_reshaping}. Moreover the variance of $\mu^+(t)$ is at most $1$ because it can only take values $\pm 1$.
We assume that $\varepsilon_{\mathrm{SPAM}}\leq 1/(3\sqrt{2})$, and choose $r_2$ to be $r_2= \mathcal{O}(M^2 t^2)$ so that 
\[
\frac{4M^2 t^2}{r_2} < \frac{1}{3\sqrt{2}}.
\]
Note that $\tau$ and $r_2$ are related through $\tau=t/r_2$.
Then we have
\[
|\mathbb{E}[\nu_s^+(t)]-\cos(2\mu_s t)| < \frac{2}{3\sqrt{2}}.
\]
We then take $54$ independent samples of $\nu_s^+(t)$ and average them, denoting the sample average by $X_s(t)$. 

By Chebyshev's inequality, this ensures
\[
\Pr[|X_s(t)-\mathbb{E}[\nu_s^+(t)]|\geq 1/(3\sqrt{2})]=\Pr[|X_s(t)-\mathbb{E}[X_s(t)]|\geq 1/(3\sqrt{2})]\leq \frac{1}{54\times 1/(3\sqrt{2})^2}=\frac{1}{3}.
\]
Therefore combining the above with \eqref{eq:phase_estimation_experiment_expectation_deviation} we have
\[
\Pr[|X_s(t)-\cos(2\mu_s t)|\geq 1/\sqrt{2}]\leq 1/3.
\]

This guarantees estimating $\cos(2\mu_s t)$ to a constant $1/\sqrt{2}$ accuracy with at least $2/3$ probability, which gives us the $X_s(t)$ required in the robust frequency estimation protocol in Theorem~\ref{thm:robust_frequency_estimation}. The $Y_s(t)$ in Theorem~\ref{thm:robust_frequency_estimation} can be similarly obtained by using $O^-_s$ as the observable. 
We also note that because $|2\mu_s|\leq 2$, we can set $A=2$ in Theorem~\ref{thm:robust_frequency_estimation}. Therefore we can state the following for the phase estimation experiment (using the notation of $\mathcal{K}_{P_s}$ as defined in \eqref{eq:random_unitaries_ensemble}:

\begin{theorem}[\Cref{thm:robust_frequency_estimation} for the Hamiltonian learning task]
    \label{thm:phase_estimation_experiment}
    We assume that the quantum system is evolving under a Hamiltonian $H$ with $M$ terms with the absolute value of the coefficient of each term bounded by $1$.
    With $N_{\mathrm{exp}}$ independent non-adaptive $(s_l,t_l,\tau_l)$-phase estimation experiments (\Cref{def:phase_estimation_experiment}), $j=1,2,\cdots,L$ ($L$ is the number of coefficients of different Pauli terms in $H$ to be estimated), with the SPAM error (\Cref{defn:SPAM_err}) satisfying $\varepsilon_{\mathrm{SPAM}}\leq 1/(3\sqrt{2})$, we can obtain an estimate $\hat{\theta}$ such that
    \[
    \Pr[|\hat{\theta}-2\mu_s|\geq \eta] \leq q,
    \]
    where $\mu_s$ are eigenvalues of the effective Hamiltonian $H_{\mathrm{eff}} = \mu_sP_s$ as defined in \eqref{eq:effective_hamiltonian_terms}.
    In the above $N_{\mathrm{exp}}$, $\{t_l\}$, and $\{\tau_l\}$ satisfy
    \[
    N_{\mathrm{exp}} = \mathcal{O}(\log(1/\eta)(\log(1/q)+\log\log(1/\eta))),
    \]
    \begin{equation}
    \begin{aligned}
        T=\sum_{l=1}^{N_{\mathrm{exp}}} t_l=\mathcal{O}\left(\frac{1}{\eta}(\log(1/q)+\log\log(1/\eta))\right),\quad \max_l t_l=\mathcal{O}(1/\eta),
    \end{aligned}
    \end{equation}
    and $\tau_l = \Omega(1/(M^2 t_l))$.
\end{theorem}

\subsection{The complexity of \texorpdfstring{$\mathcal{A}^{II}_j$}{A II j}}
\label{sec:complexities_of_A_2}
In $\mathcal{A}^{II}_j$, we learn each coefficient in $\mathfrak{S}_{L_j}$ to accuracy $\epsilon$. There are $|\mathfrak{S}_{L_j}|$ terms to be estimated by the robust frequency estimation as described in \Cref{sec:robust_frequency_estimation_details}. By \eqref{eq:total_evolution_time_RPE}, the total evolution time under $H$ in each robust frequency estimation instance is
\begin{equation}
    t_{j,l} = \mathcal{O}((1/\epsilon)(\log(1/\epsilon)+\log\log(2^{-j}/\epsilon))),\quad l=\{0,1,2,\ldots,|\mathfrak{S}_{L_j}|\}.
\end{equation}
Therefore, the total evolution time under $H$ in $\mathcal{A}^{II}_j$ is:
\begin{equation}
    T_{2,j} = \sum_{l=1}^{|\mathfrak{S}_{L_j}|} t_{j,l} = |\mathfrak{S}_{L_j}|\cdot  \mathcal{O}((1/\epsilon)(\log(1/\epsilon)+\log\log(2^{-j}/\epsilon))) = \mathcal{O}\left( \frac{M^2\log(M/\delta_j)}{\epsilon}(\log(1/\epsilon)+\log\log(2^{-j}/\epsilon))\right).
\end{equation}

\section{Total evolution time complexity}
\label{sec:total_time_complexity}
Combining the results in \Cref{sec:complexities_of_A_1} and \Cref{sec:complexities_of_A_2}, the total evolution time in $\mathcal{A}_j$ is 
\begin{equation}
    T_j = T_{1,j} + T_{2,j} = \mathcal{O}\left(2^j M\log(M/\delta_{j}) + \frac{M^2\log(M/\delta_j)}{\epsilon}(\log(1/\epsilon)+\log\log(2^{-j}/\epsilon))\right).
\end{equation}
Thus, the total number of experiments is
and the total evolution time complexity of our protocol is
\begin{equation}
\label{eq:total_number_of_exp}
    L = \sum_{j=0}^{\left\lceil\log_2(1/\epsilon)\right\rceil-1} L_j = \mathcal{O}\left(M^2\log(M/\delta)\log(1/\epsilon)\right),
\end{equation}
and the total evolution time complexity of our protocol is
\begin{equation}
\label{eq:total_evolution_time_complexity}
    T = \sum_{j=0}^{\left\lceil\log_2(1/\epsilon)\right\rceil-1} T_j = \mathcal{O}\left(\frac{M\log(M/\delta)\log(1/\epsilon)}{\epsilon} + \frac{M^2\log(M/\delta)\log(1/\epsilon)}{\epsilon}(\log(1/\epsilon)+\log\log(1/\epsilon))\right),
\end{equation}
where $\delta = \min\{\delta_j\}$.

\section{Alternative single-copy product state input approach \texorpdfstring{$\mathcal{A}^{I'}_j$}{A I' j}}
\label{sec:single-copy}
We here provide the proof for the alternative approach of realizing the large-term identification in $\mathcal{A}^{I}_j$ which requires no ancillary qubit and only product state input. With this approach, the entire Hamiltonian learning protocol requires no ancillary qubits, only product state input, and still achieves Heisenberg-limited scaling. This approach replaced the step of identifying terms in $\mathfrak{S}_j$ by using the population recovery protocol for estimating the Pauli error rate of a general quantum channel in \cite{Flammia2021paulierror}. Instead of using this protocol to directly learn the Pauli error rates to a high accuracy, we learn them to a low accuracy and only take those large terms over a threshold to be estimated to high accuracy in $\mathcal{A}^{II}_j$.

\subsection{Pauli error rate for time-evolution channel}
\label{sec:Pauli_error_rate}
Consider a general quantum channel $\Lambda$, one can write it in its Kraus operator expression as in\eqref{eq:Kraus_operator}.
\begin{defn}[Pauli twirling of a quantum channel, Definition~28 in \cite{Flammia2021paulierror}]
\label{def:Pauli_twirl}
    Let $\Lambda$ denote an arbitrary $n$-qubit quantum channel. Its Pauli twirl $\Lambda_P$ is the $n$-qubit quantum channel defined by
    \begin{equation}
        \Lambda_P(\rho) = \mathop{\mathbb{E}}_{\sigma_T\in\{I, \sigma_x, \sigma_y, \sigma_z\}}\left[ \sigma_T^\dagger \left( \Lambda \left( \sigma_T \rho \sigma_T^\dagger \right) \right) \sigma_T \right].
    \end{equation}
\end{defn}

\begin{fact}[Pauli error rates of general quantum channels]
    Suppose we write $K_j$ for the Kraus operators of $\Lambda$, so $\Lambda\rho = \sum_j  K_j\rho K_j^\dagger$ . Further suppose that $K_j$ is represented in the Pauli basis as in \eqref{eq:Pauli_expansion_of_Kraus}. Then $\Lambda$’s  Pauli error rates are given by $p(k) = \sum_j |\alpha_{j,k}|^2$
\end{fact}

Consider the time evolution channel under Hamiltonian $\tilde{H}_j$ as in \Cref{sec:Bell_sampling}, there is only one Kraus operator in this Pauli channel $K=e^{-i\tilde{H}_jt}$. We obtain the expansion of this Kraus operator in the Pauli basis by considering its Taylor expansion as in \Cref{eq:Taylor_expansion_terms}. Similar to \Cref{eq:upper_bound_for_t}, we set $t < \frac{1}{CM}$. Then, by \Cref{sec:lower_bound_prob} the Pauli error rate of a term $P_s$ where $s\in\mathfrak{S}_j$ is lower bounded by $\widetilde{\Omega}(\frac{1}{4C^2M^2} - \frac{1}{C^3M^2}-\varepsilon_{\mathrm{SPAM}})$.

\subsection{Estimating Pauli error rates via population recovery}
\label{sec:population_recovery}
\begin{theorem}[Theorem~1 in \cite{Flammia2021paulierror}]
\label{thm:intro-pop}
    There is a learning algorithm that, given parameters $0 < \delta, \epsilon_{1} < 1$, as well as access to an $n$-qubit channel with Pauli error rates~$p$, has the following properties:
    \begin{itemize}
        \item It prepares $m = \mathcal{O}(1/\epsilon_1^2)\cdot \log(\frac{n}{\epsilon_{1} \delta})$ unentangled $n$-qubit pure states, where each of the $mn$ $1$-qubits states is chosen uniformly at random from $\{\ket{0}, \ket{1}, \ket{+}, \ket{-}, \ket{i}, \ket{-i}\}$;
        \item It passes these $m$ states through the Pauli channel.
        \item It performs unentangled measurements on the resulting states, with each qubit being measured in either the $\{\ket{0}, \ket{1}\}$-basis, the $\{\ket{+}, \ket{-}\}$-basis, or $\{\ket{i}, \ket{-i}\}$-basis.
        \item It performs an $O(mn/\epsilon_1)$-time classical post-processing algorithm on the resulting $mn$ measurement outcome bits.
        \item It outputs hypothesis Pauli error rates $\widehat{p}$ in the form of a list of at most $\frac{4}{\epsilon_1}$ pairs $(k, \widehat{p}(k))$, with all unlisted $\widehat{p}$ values treated as~$0$.
    \end{itemize}
    The algorithm's hypothesis $\widehat{p}$ will satisfy $\|\widehat{p} - p\|_\infty \leq \epsilon_1$ except with probability at most~$\delta$.
\end{theorem}

\begin{fact}[Fact~29 in \cite{Flammia2021paulierror}]
    Applying the population recovery protocol in \cite{Flammia2021paulierror} to a general quantum channel $\Lambda$ can learn its Pauli error rates with the properties mentioned in \Cref{thm:intro-pop}.
\end{fact}

\begin{lemma}[Lower bound for Pauli error rates]
    For any term $P_s$ where $s\in\mathfrak{S}_j$, the Pauli error rate of the time evolution channel under $e^{-i\tilde{H}_jt}$ realized by Trotterization is lower bounded by
    \begin{eqnarray}
        \gamma_j = \widetilde{\Omega}\left(\frac{1}{4C^2M^2}-\frac{1}{C^3M^2}-\frac{2^{2j}}{C^2r}-\varepsilon_{\text{SPAM}}\right)
    \end{eqnarray}
    for $t = \Theta(\frac{1}{CM})$, and $\varepsilon_{\text{SPAM}}$ is the SPAM error in experiments. 
\end{lemma}
\begin{proof}
    The proof is exactly the same as the proof of \Cref{lem:lower_bound_for_prob_dist} since the analysis for Taylor expansion in \Cref{sec:Pauli_error_rate} is the same with the analysis in \Cref{sec:Bell_sampling}.
\end{proof}

Let the precision of learning Pauli errors to be 
\begin{equation}
    \epsilon_1 = (\frac{1}{2}-c')\frac{1}{CM},
\end{equation}
where $c'$ satisfies that $(\frac{1}{2}-c')\frac{1}{CM} < \frac{\gamma_j}{2}$ and is only introducd for simplifying the expression, it would not contribute to the leading terms in the following complexity analysis. Then, for a term $P_s$ where $s\in\mathfrak{S}_j$ with the absolute value its rescaled coefficient in $\tilde{H}_j$ larger than $1/2$, by applying the population recovery protocol in \Cref{thm:intro-pop}, the learned coefficient recovered from the estimated Pauli error rate will be lower bounded by $1/4$ with probability at least $1-\delta$. By seeting the threshold to be $\frac{1}{2}\gamma_j$ and pass all terms over this threshold to be estimated to high accuracy in $\mathcal{A}^{II}_j$, we can cover all terms in $\mathfrak{S}_j$ with high probability.

\subsection{Complexity of \texorpdfstring{$\mathcal{A}^{I'}_j$}{A I' j}\label{sec:complexity_single_copy}}
By \Cref{thm:intro-pop}, the number of input states $m = \mathcal{O}(1/\epsilon_1^2)\cdot \log(\frac{n}{\epsilon_1\delta})$, and from Appendix~\ref{sec:population_recovery}, we set $\epsilon_1 = \frac{1}{2}\gamma_j$, then
\begin{equation}
\label{eq:num_of_input_stabilizer_product_state}
    m = \mathcal{O}\left(\frac{C^4M^4}{(\frac{1}{2}-c')^2}\log(\frac{CMn}{\delta})\right).
\end{equation}
For each input, the total evolution time under $H$ is
\begin{equation}
    t'_{j,1} = t/2^{-j} = \Theta\left(\frac{1}{CM2^{-j}}\right),
\end{equation}
so the total evolution time of $\mathcal{A}^{I'}_j$ is
\begin{equation}
    T'_{1,j} = m\cdot t'_{j,1} = \mathcal{O}\left(\frac{C^3M^3}{2^{-j}(\frac{1}{2}-c')}\log(\frac{CMn}{\delta})\right) \approx \mathcal{O}\left(\frac{C^3M^3}{2^{-j}}\log(\frac{CMn}{\delta})\right).
\end{equation}
Meanwhile, by \Cref{thm:intro-pop}, $\mathcal{A}^{I'}_j$ requires a $T^C_j$-time classical post-processing algorithm on the resulting $mn$ measurement outcome bits
\begin{equation}
    T^C_j = \mathcal{O}\left(\frac{mn}{\epsilon_1}\right) = \mathcal{O}\left(\frac{C^5M^5n}{(\frac{1}{2}-c')^2}\log\left(\frac{CMn}{\delta}\right)\right).
\end{equation}

\subsection{Totoal evolution time complexity of the single-copy product state input Hamiltonian learning protocol}
\label{sec:total_evo_time_for_single_copy}
By replacing $\mathcal{A}^{I}_j$ by $\mathcal{A}^{I'}_j$, we obtain the single-copy product state input Hamiltonian learning protocol. The total evolution time complexity is
\begin{equation}
    T' = \sum_{j=0}^{\left\lceil\log_2(1/\epsilon)\right\rceil-1} \left(T'_{1,j}+T_{2,j}\right) = \mathcal{O}\left(\frac{C^3M^3}{\epsilon}\log(\frac{CMn}{\delta})\log(1/\epsilon) + \frac{M^2\log(M/\delta)\log(1/\epsilon)}{\epsilon}(\log(1/\epsilon)+\log\log(1/\epsilon))\right),
\end{equation}
and the total classical post-processing time is
\begin{eqnarray}
    T^C = \sum_{j=0}^{\left\lceil\log_2(1/\epsilon)\right\rceil-1} T^C_j = \mathcal{O}\left(C^5M^5n\log\left(\frac{CMn}{\delta}\right)(\left\lceil\log_2(1/\epsilon)\right\rceil-1)\right)
\end{eqnarray}

\section{Proof for the trade-off\label{app:lower}}
In this section, we provide the proof for the lower bound in \Cref{thm:lower_informal}, which is formulated as follows. 
\begin{theorem}\label{thm:lower}
For any protocol that can be possibly adaptive, biased or unbiased, and possibly ancilla-assisted with total evolution time $T$ and at most $\mathcal{L}$ discrete quantum controls per experiment (as shown in Figure S1),  there exists some parametrized Hamiltonian $H$ that requires $T = \Omega(\mathcal{L}^{-1}\epsilon^{-2})$ to estimate the coefficient of the unknown Hamiltonian $H$ within additive error $\epsilon$.

\end{theorem}

We consider the problem of learning a Hamiltonian of $M$ terms 
\begin{align}\label{eq:param_ham}
H(\bm{\mu})=\sum_{s=1}^M\mu_s P_s
\end{align} 
with $\bm{\mu}=(\mu_1\ldots,\mu_{M})$ and $\abs{\mu_s}\leq 1$ for any $s=1,...,M$. 

\begin{figure}[ht]
    \centering
    \includegraphics[width=0.99\textwidth]{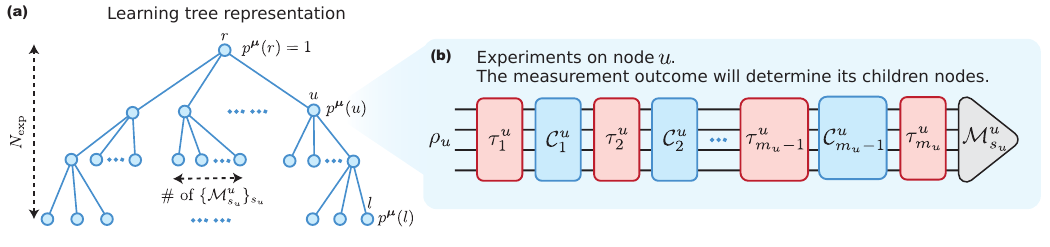}
    \caption{(a) Learning tree representation. Each node $u$ represents an experiment. Starting from the root experiment $r$, the number of child nodes depends on the possible POVM measurements. The transition probabilities are determined by \Cref{app_eq:transition_prob}. After $N_\text{exp}$ experiments, one arrives at the leaves $l$ of the learning tree. (b) In each node u, the learning model prepares an arbitrary state $\rho_u$, applies discrete control channels $\mathcal{C}^u_{k}$, and queries real-time Hamiltonian evolutions with time $\tau_k^{u}$ multiple times. The protocols can also incorporate ancilla qubits (quantum memory).
 }
    \label{fig:learn_lower_model}
\end{figure}

The high-level strategy of our proof is to combine the learning tree framework~\cite{learningtree} equipped with the martingale trick~\cite{ComplexityNISQ,2023arXiv230914326C,OptimalTradeOff} and quantum Fisher information.

Here, we generalize any experimental settings with discrete quantum controls and sequential queries to the Hamiltonian as the following theoretical model. 
Given a Hamiltonian $H$ with coefficient vector $\bm{\mu}$ defined in \eqref{eq:unknown_Hamiltonian}, the protocol performs multiple experiments as shown in \Cref{fig:learn_lower_model} and measures at the end of each experiment.
In each experiment (indexed by $u$), the protocol prepares an input state $\rho_u$ (possibly with ancilla qubits) and queries the Hamiltonian $m_u$ times with evolution time $\tau_1^u,...,\tau_{m_u}^u$. We can also apply discrete quantum control unitaries $\mathcal{C}_1^u,...,\mathcal{C}_{m_u-1}^{u}$ between every two neighboring queries.
The protocol can be adaptive in the sense that it can dynamically decide how to prepare the input state, query the real-time evolutions, perform quantum controls, and measure the final state based on the history of the previous experiments.
We denote the summation of the evolution times for querying the Hamiltonian throughout \emph{all} experiment to be the total evolution time $T$ and the maximal number of discrete quantum controls in \emph{one} experiment as $\mathcal{L}$.

Formally, we describe the protocol scheme in \Cref{fig:learn_lower_model} as the learning tree model~\cite{learningtree,OptimalTradeOff} as follows:
\begin{defn}[Learning tree representation] 
Given a Hamiltonian $H(\bm{u})$ as defined in \eqref{eq:param_ham} with $M$ terms, a learning protocol using discrete control channels and queries to the Hamiltonian real-time evolutions as shown in \Cref{fig:learn_lower_model} can be represented as a rooted tree $\mathcal{T}$ of depth $N_{\exp}$ (corresponding to $N_{\exp}$ measurements) with each node representing the measurement outcome history so far. In addition, the following conditions are satisfied:
\begin{itemize}
    \item We assign a probability $p^{\bm{\mu}}(u)$ for any node $u$ on the tree $\mathcal{T}$. The probability assigned to the root $r$ is $p^{\bm{\mu}}(r)=1$.
    \item At each non-leaf node $u$, we input the state $\rho_u$. By convexity, we can assume $\rho_u=\ket{\psi_u}\bra{\psi_u}$ a pure state. We then query $m_t$ rounds of real-time Hamiltonian evolution of time $\tau^u_{1},...,\tau^u_{m_t}$, interleaved with $m_t$ discrete control channels $\mathcal{C}_1^{u},...,\mathcal{C}_{m_t}^u$ (the last one, $\mathcal{C}_{m_t}^u$, can be absorbed into the measurement as shown in the figure). We then perform a POVM $\{\mathcal{M}_{s_u}^u\}_{s_u}$, which by the fact that any POVM can be simulated by rank-$1$ POVM~\cite{learningtree}, can be assumed to be a rank-$1$ POVM as $\{w^u_{s_u}\ket{\phi_{s_u}^u}\bra{\phi_{s_u}^u}\}_{s_u}$ with $\sum w^u_{s_u}=2^n$, and get classical output $s_u$. The child node $v$ corresponding to the classical outcome $s_u$ of the node $u$ is connected through the edge $e_{u,s_u}$. The probability associated with $v$ is given by:
    \begin{align}
    p^{\bm{\mu}}(v)\triangleq p^{\bm{\mu}}(u)\cdot\tr\left(w^u_{s_u}\ket{\phi_{s_u}^u}\bra{\phi_{s_u}^u}\cdot \mathcal{U}_{m_u}\mathcal{C}_{m_u-1}\mathcal{U}_{m_u-1}\cdots \mathcal{U}_2\mathcal{C}_1 \mathcal{U}_1(\ket{\psi_u}\bra{\psi_u})\right),\label{app_eq:transition_prob}
    \end{align}
    where $\mathcal{U}_i=e^{-iH(\bm{\mu})\tau_{i}}$ is the unitary real-time evolution of $H(\bm{\mu})$ for time $\tau_i$.
    \item Each root-to-leaf path is of length $N_{\exp}$. For a leaf node $\ell$, $p^{\bm{\mu}}(\ell)$ is the probability of reaching this leaf $\ell$ at the end of the learning protocol. The set of leaves is denoted as $\text{leaf}(\mathcal{T})$.
\end{itemize}
\end{defn}

We consider a point-versus-point distinguishing task between two cases for a hyperparameter $\bm{\mu}$:
\begin{itemize}
    \item The Hamiltonian for the real-time evolution is exactly $H(\bm{\mu})$.
    \item The Hamiltonian for the real-time evolution is actually $H(\bm{\mu}+\delta\bm{\mu})$ with $\norm{\delta\bm{\mu}}_\infty=3\epsilon$.
\end{itemize}
The goal is to distinguish which case is happening. 
It is straightforward to see that given an algorithm that can learn the Hamiltonian, we can solve this distinguishing problem. 

Now, we consider how hard this distinguishing problem is. 
In the learning tree framework, the necessary condition for distinguishing between these two cases according to Le Cam's two-point method~\cite{yu1997assouad} is that the total variation distance between the probability distributions on the leaves of the learning trees for the two cases is large. 
Quantitatively,
\eqs{
\frac{1}{2}\sum_{\ell\in\text{leaf}(\mathcal{T})}\abs{p^{\bm{\mu}+\delta\bm{\mu}}(\ell)-p^{\bm{\mu}}(\ell)}=\Theta(1). \label{app_eq:leaf_prob_diff}
}

Using the martingale likelihood ratio argument developed by a series of works~\cite{ComplexityNISQ,2023arXiv230914326C,2021arXiv211207646C}, the above lower bound can be converted to Lemma 7 of Ref.~\cite{OptimalTradeOff}: If there is a $\Delta>0$ such that 
\begin{align}\label{eq:martingale}
\mathbb{E}_{s_u\sim p^{\bm{\mu}}(s_u|u)}[(L_{\delta\bm{\mu}}(u,s_u)-1)^2]\leq\Delta,
\end{align}
where
\begin{align}
L_{\delta\bm{\mu}}(u,s_u)=\frac{p^{\bm{\mu}+\delta\bm{\mu}}(s_u|u)}{p^{\bm{\mu}}(s_u|u)},
\end{align}
then $N_{\exp}\geq\Omega(1/\Delta)$. The physical intuition behind this result is that when the difference between the probability  $p^{\bm{\mu}}(s_u|u)$  and the perturbed probability  $p^{\bm{\mu}+\delta \bm{\mu}}(s_u|u)$  is very small across all levels of the learning tree (\Cref{eq:martingale}), a very deep tree is required—i.e.,  $N_{\exp} \geq \Omega(1/\Delta)$ —to ensure that the measurement distributions at the leaf nodes are distinct (\Cref{app_eq:leaf_prob_diff}).

Now, we focus on this quantity
\begin{align}
\mathbb{E}_{s_u\sim p^{\bm{\mu}}(s_u|u)}[(L_{\delta\bm{\mu}}(u,s_u)-1)^2]
\end{align}
for a specific $u$. Recall that $\norm{\delta\bm{\mu}}_{\infty}=3\epsilon$, we thus have $\norm{\delta\bm{\mu}}_2\leq 3\sqrt{M}\epsilon$. We can thus compute this quantity by extending it to the first-order term as
\begin{align}
\begin{split}
\mathbb{E}_{s_u\sim p^{\bm{\mu}}(s_u|u)}[(L_{\delta\bm{\mu}}(u,s_u)-1)^2]&=\mathbb{E}_{s_u\sim p^{\bm{\mu}}(s_u|u)}\left[\left(\frac{p^{\bm{\mu}+\delta\bm{\mu}}(s_u|u)}{p^{\bm{\mu}}(s_u|u)}-1\right)^2\right]\\
&\simeq\mathbb{E}_{s_u\sim p^{\bm{\mu}}(s_u|u)}\left[\left(\frac{p^{\bm{\mu}}(s_u|u)+\nabla_{\bm{\mu}}p^{\bm{\mu}}(s_u|u)\cdot\delta\bm{\mu}}{p^{\bm{\mu}}(s_u|u)}-1\right)^2\right]\\
&=\mathbb{E}_{s_u\sim p^{\bm{\mu}}(s_u|u)}\left[\left(\frac{\nabla_{\bm{\mu}}p^{\bm{\mu}}(s_u|u)\cdot\delta\bm{\mu}}{p^{\bm{\mu}}(s_u|u)}\right)^2\right]\\
&=\delta\bm{\mu}^{\top}I^{(C)}_{\bm{\mu},u}\delta\bm{\mu}\\
&\leq\max_{\bm{v}:\norm{\bm{v}}_2\leq 3\sqrt{M}\epsilon}\bm{v}^\top I^{(C)}_{\bm{\mu},u}\bm{v},
\end{split}
\end{align}
where 
\begin{align}
\begin{split}
\left[I_{\bm{\mu},u}^{(C)}\right]_{a,a'}&=\mathbb{E}_{s_u\sim p^{\bm{\mu}}(s_u|u)}\left[\left(\frac{\frac{\partial p^{\bm{\mu}}(s_u|u)}{\partial\lambda_a}}{p^{\bm{\mu}}(s_u|u)}\right)\left(\frac{\frac{\partial p^{\bm{\mu}}(s_u|u)}{\partial\lambda_{a'}}}{p^{\bm{\mu}}(s_u|u)}\right)\right]\\
&=\mathbb{E}_{s_u\sim p^{{\bm{\mu}}}(s_u|u)}\left[\left(\frac{\partial \ln p^{\bm{\mu}}(s_u|u)}{\partial\lambda_a}\right)\left(\frac{\partial \ln p^{\bm{\mu}}(s_u|u)}{\partial\lambda_{a'}}\right)\right]
\end{split}
\end{align}
is the classical Fisher information matrix~\cite{PhysRevD.23.357,PhysRevLett.72.3439,BRAUNSTEIN1996135}. Note that quantum Fisher information $\bm{v}^{\top}I^{(Q)}_{{\bm{\mu}},u}\bm{v}$ is the supremum of classical Fisher information $\bm{v}^{\top}I^{(C)}_{{\bm{\mu}},u}\bm{v}$ over all possible measurements or observables~\cite{Paris2009Quantum}. By Lemma 17 of Ref.~\cite{dutkiewicz2024advantage}, assuming $H({\bm{\mu}})$ can be diagonalized as $H({\bm{\mu}})=W({\bm{\mu}})^\dagger D({\bm{\mu}}) W({\bm{\mu}})$ with unitary $W$ and diagonal $D$, we have the quantum Fisher information satisfies
\begin{align}
\max_{\bm{v}:\norm{\bm{v}}_2\leq 3\sqrt{M}\epsilon}\bm{v}^\top I^{(Q)}_{{\bm{\mu}},u}\bm{v}\leq M\epsilon^2\left(\min\left\{t_u\norm{\partial_{\bm{\mu}} D}_{\text{HS}}+2m_u\norm{\partial_{\bm{\mu}} W}_{\text{HS}},t_u\norm{\partial_{\bm{\mu}} H}_{\text{HS}}\right\}\right)^2,
\end{align}
where $t_u=\sum_{i=1}^{m_u}\tau_i^u$ and $\norm{\cdot}_{\text{HS}}$ is the Hilbert Schmidt norm (spectral norm). By applying Eq. \eqref{eq:martingale}, we have
\begin{align}
N_{\exp}=\Omega\left(\left(M\epsilon^2\left(\min\left\{2m_u\norm{\partial_{\bm{\mu}} W}_{\text{HS}},t_u\norm{\partial_{\bm{\mu}} H}_{\text{HS}}\right\}\right)^2\right)^{-1}\right)
\end{align}
for any protocols that can solve this distinguishing problem with a high probability. The RHS is maximized when $t_u\leq2m_u\norm{\partial_{\bm{\mu}} W}_{\text{HS}}/\norm{\partial_{\bm{\mu}} H}_{\text{HS}}$. Denote $\mathcal{L}=\max_u m_u$ as the maximal number of discrete quantum controls in one experiment and $T=\sum_{u=1}^{N_{\text{exp}}} t_u$ as the total evolution time, we have
\begin{align}
2M\epsilon^2\mathcal{L}T\norm{\partial_{\bm{\mu}} W}_{\text{HS}}\norm{\partial_{\bm{\mu}} H}_{\text{HS}}=\Theta(1),
\end{align}
which indicates that any protocol that can solve this distinguishing problem with a high probability requires $T=\Omega(\mathcal{L}^{-1}\epsilon^{-2})$ as claimed in \Cref{thm:lower}.

\section{Details of the numerical simulation}
\label{sec:numerics_details}
In this appendix, we provide additional details on our numerical simulations. For structure learning, we perform 2000 measurement shots, while for coefficient learning, we use 1000 simulated measurements throughout. A key advantage of our Heisenberg-limited learning algorithm is that the number of required measurements does not scale with the target learning accuracy. This allows us to use a constant number of measurements even when estimating small coefficients. For short-time evolutions, we set the evolution time to $\tau = 0.001$, and we find that the results are not particularly sensitive to the precise choice of $\tau$.

For the first example with learning disordered XY model, we perform the hierarchical learning with $J=1$. The total evolution time in structure learning is chosen as $T^{j=1}=0.3$, and the learning precision in the coefficient learning is set to be $\epsilon=0.005$.

For the second example: learning long-range interaction decay in a static Rydberg Hamiltonian, we perform hierarchical learning up to level $J = 5$. The physical Hamiltonian takes the form:
\eqs{
H/\hbar = &\frac{\Omega}{2} \sum_{l} \left( |g_l\rangle \langle r_l| + |r_l\rangle \langle g_l| \right) - \Delta \sum_{l} |r_l\rangle \langle r_l| \\
&+ \sum_{j<l} V_{jl} |r_j\rangle \langle r_j| \otimes |r_l\rangle \langle r_l| ,
}
where $V_{jl} = \mathrm{C} / |\vec{x}_j - \vec{x}_l|^6$ describes the van der Waals interaction between atoms $j$ and $l$. We use $\mathrm{C} = 862{,}690 \times 2\pi$ MHz·$\mu$m$^6$, corresponding to the interaction strength for Rydberg states $70S_{1/2}$ of $^{87}\text{Rb}$ atoms. The atoms are equally spaced with a separation of $10\mu\text{m}$, and the drive parameters are set to $\Omega = 1.5/2\pi$ MHz and $\Delta = -4/2\pi$ MHz. Due to the $1/r^6$ decay of the van der Waals interaction, the interaction strength rapidly diminishes with distance. We placed all atoms in a chain separated equally with $d=10\mu\text{m}$. For each level $j$ in the hierarchical learning protocol, we set the evolution time $T^{j}$ and learning precision $\epsilon^{j}$ as follows:
\begin{enumerate}
    \item Level 1: $T^{j=1}=0.08\mu\text{s}$ and $\epsilon^{j=1}=0.005 \text{ rad}/\mu\text{s}$
    \item Level 2: $T^{j=2}=0.2\mu\text{s}$ and $\epsilon^{j=2}=0.005\text{ rad}/\mu\text{s}$
    \item Level 3: $T^{j=3}=6.0\mu\text{s}$ and $\epsilon^{j=3}=0.001\text{ rad}/\mu\text{s}$
    \item Level 4: $T^{j=4}=60.0\mu\text{s}$ and $\epsilon^{j=4}=0.0001\text{ rad}/\mu\text{s}$
    \item Level 5: $T^{j=5}=1200.0\mu\text{s}$ and $\epsilon^{j=5}=0.00001\text{ rad}/\mu\text{s}$
\end{enumerate}
The Pauli terms learned at each level are summarized in \Cref{table:rydberg}.

For the final example: learning the effective Hamiltonian generated by a time-dependent driven pulse, we consider a Rydberg atom array with the same interaction strength as before, but with an interatomic spacing of $d = 8.9\mu\text{m}$. The system is driven by a time-dependent laser pulse, as shown in \Cref{fig:example2}(a). We do the Hamiltonian learning with the information that this pulse engineered a unitary that is close to 
\eqs{
U\approx e^{-i\theta H_{\text{ZXZ}}},
}
with $\theta=0.1$ and $H_{\text{ZXZ}}=\sum_{i}Z_{i}X_{i+1}Z_{i+2}$. Usually, this information is available because engineered pulses are optimized to reach a certain target with the target being known. Thus, the unitary implemented by the pulse can be viewed as a short-time evolution under an unknown effective Hamiltonian that is close to $H_{\text{ZXZ}}$. Since we lack prior knowledge about the exact form of this effective Hamiltonian, the ansatz-free learning protocol is particularly well-suited for this task.

In the structure learning stage, we prepare Bell pairs, apply the time-dependent Hamiltonian to one half of each pair, and then perform Bell-basis measurements. We use 1000 measurement shots for structure learning. Once the structure is identified, we treat $U$ as a short-time evolution with $\theta = 0.1$ under an unknown effective Hamiltonian, and proceed with coefficient learning using Hamiltonian reshaping. For this step, we also use 1000 measurement shots, with the target learning accuracy set to $\epsilon = 0.01$.

\begin{table}[htbp]
\caption{Hierarchical learning of static Rydberg atom Hamiltonian\label{table:rydberg}}
\vspace{5mm}
\begin{ruledtabular}
\begin{tabular}{lcc}
$j$th & Pauli String & Value \\
\hline
1 & IIIZZ & $1.343 \pm 0.036$ \\
1 & IIZZI & $1.331 \pm 0.047$ \\
1 & IZZII & $1.359 \pm 0.044$ \\
1 & ZZIII & $1.355 \pm 0.025$ \\
1 & IIIIZ & $0.618 \pm 0.029$ \\
1 & IIIZI & $-0.733 \pm 0.043$ \\
1 & XIIII & $0.750 \pm 0.019$ \\
1 & IZIII & $-0.728 \pm 0.043$ \\
1 & IIIXI & $0.741 \pm 0.026$ \\
1 & IXIII & $0.746 \pm 0.029$ \\
1 & ZIIII & $0.624 \pm 0.036$ \\
1 & IIXII & $0.748 \pm 0.033$ \\
1 & IIIIX & $0.751 \pm 0.013$ \\
1 & IZZXI & $0.013 \pm 0.059$ \\
\midrule
2 & IZIII & $-0.753 \pm 0.001$ \\
\midrule
3 & IIZZI & $0.0245 \pm 0.0006$ \\
3 & IZIZI & $0.0211 \pm 0.0003$ \\
3 & ZIZII & $0.0212 \pm 0.00007$ \\
3 & IIZIZ & $0.0214 \pm 0.0008$ \\
3 & IZZXI & $-0.0130 \pm 0.0006$ \\
3 & IIIZZ & $0.0123 \pm 0.0008$ \\
3 & IIIXI & $0.0088 \pm 0.0005$ \\
3 & IZZII & $-0.0040 \pm 0.0005$ \\
3 & ZIIZI & $0.0019 \pm 0.0007$ \\
\midrule
4 & IZIII & $-0.00525 \pm 0.00007$ \\
4 & IXIII & $0.00390 \pm 0.00008$ \\
4 & IIIIZ & $0.00374 \pm 0.00004$ \\
4 & ZIIII & $-0.00223 \pm 0.00001$ \\
4 & IZIIIZ & $0.00186 \pm 0.00013$ \\
4 & IIIYI & $0.00199 \pm 0.000005$ \\
4 & IIIIX & $-0.00149 \pm 0.00006$ \\
4 & IXIIZ & $-0.000012 \pm 0.00012$ \\
4 & ZZIII & $0.00022 \pm 0.00009$ \\
4 & IYIII & $0.00006 \pm 0.00011$ \\
4 & IXIIX & $-0.000013 \pm 0.00011$ \\
4 & ZXIII & $0.000007 \pm 0.00011$ \\
4 & IYIIZ & $-0.000001 \pm 0.00012$ \\
4 & IYYII & $-0.000012 \pm 0.00014$ \\
\midrule
5 & ZIIIZ & $0.000331 \pm 0.000001$ \\
\end{tabular}
\end{ruledtabular}
\end{table}

\end{appendix}

\end{document}